\documentclass{llncs}

\usepackage{amsmath}
\usepackage{amsfonts}
\usepackage{amssymb}
\usepackage{verbatim}
\usepackage{pst-all}
\usepackage[dvips,xdvi]{graphicx}
\usepackage{graphics}

\usepackage{stmaryrd}
\usepackage{verbatim}

\usepackage{floatflt}

\begin{document}

\title{Finite Optimal Control for Time-Bounded Reachability in CTMDPs and Continuous-Time Markov Games
\thanks{This work was partly supported by the German Research
Foundation (DFG) as part of the Transregional Collaborative Research
Center ``Automatic Verification and Analysis of Complex Systems''
(SFB/TR 14 AVACS) and by the Engineering and Physical Science Research Council (EPSRC) through grant EP/H046623/1 ``Synthesis and Verification in Markov Game Structures''.}%
}



\author{Markus Rabe\inst{1} \and Sven Schewe\inst{2}}
\institute{
Universit\"at des Saarlandes
\and
University of Liverpool
}

\pagestyle{plain}


\date{Received: date / Accepted: date}


\newenvironment{proofidea}{
        \noindent {\bf Proof Idea: }}{ \qed}

\newcommand{\R}{\mathbb{R}} 
\newcommand{\N}{\mathbb{N}} 
\newcommand{\Q}{\mathbb{Q}} 


\newcommand{\locations}{{L}}
\newcommand{\act}{\mathit{Act}} 
\newcommand{\ratematrix}{\mathbf{R}}
\newcommand{\probabilitymatrix}{\mathbf{P}}

\newcommand{\nM}{\ensuremath{n_{\mathcal M}}}
\newcommand{\M}{{\ensuremath\mathcal{M}}}
\newcommand{\U}{{\ensuremath\mathcal U }}
\renewcommand{\S}{\ensuremath{\mathcal{S}}}
\newcommand{\I}{\ensuremath{\mathcal{J}}}
\newcommand{\G}{\ensuremath{\mathcal{G}}}
\newcommand{\C}{\ensuremath{\mathcal{C}}}

\newcommand{\dist}{\mathit{Dist}}
\newcommand{\paths}{\mathit{Paths}}
\newcommand{\pathsabs}{\mathit{Paths}_{\mathit{abs}}}


\hyphenation{CTMDP}
\hyphenation{DTMDP}
\hyphenation{CTMC}
\hyphenation{MDP}
\hyphenation{CTMDPs}
\hyphenation{DTMDPs}
\hyphenation{CTMCs}
\hyphenation{MDPs}

\newcommand{\prob}{\mathit{Pr}}
\newcommand{\Prob}{\mathit{PR}}
%
%

\maketitle

\begin{abstract}
We establish the existence of optimal scheduling strategies for time-bounded reachability in continuous-time Markov decision processes, and of co-optimal strategies for continuous-time Markov games.
Furthermore, we show that optimal control does not only exist, but has a surprisingly simple structure: 
The optimal schedulers from our proofs are deterministic and timed-positional, and the bounded time can be divided into a finite number of intervals, in which the optimal strategies are positional.
That is, we demonstrate the existence of \emph{finite} optimal control. 
Finally, we show that these pleasant properties of Markov decision processes extend to the more general class of continuous-time Markov games, and that both early and late schedulers show this behaviour. 
\end{abstract}


\section{Introduction}

\begin{figure}[t]
\vspace{-.3cm}
\centering
\large
\begin{pspicture}[showgrid=false](-0.7,0)(3.7,3)
	\psset{arrowsize=5pt,nodesep=0pt,arrowlength=1,linewidth=1pt}
	\cnode[](0,1){10pt}{1}
	\rput(0,1){A}
	\psline[]{->}(-.7,1.7)(-.25,1.25)
	\cnode[linewidth=0.8pt](3,1){10pt}{goal}
	\cnode[linewidth=0.8pt](3,1){8pt}{goalinner}
	\rput(3,1){C}
	\cnode[](1.5,3.3){10pt}{2}
	\rput(1.5,3.3){B}
	\ncline[]{->}{1}{2}
	\ncline[]{->}{2}{goal}
	\ncline[]{->}{1}{goal}
	\rput(0.65,2.4){4}
	\rput(0.25,1.7){$a$}
	\rput(2.25,2.55){$a$}
	\rput(1.8,0.8){2}
	\rput(0.7,0.8){\bfseries $b$}
	\rput(2.75,1.8){4}
	\rput(8.4,3.4){$t'$}
\end{pspicture}
\normalsize
\vspace{-.3cm}\includegraphics[width=0.52\columnwidth]{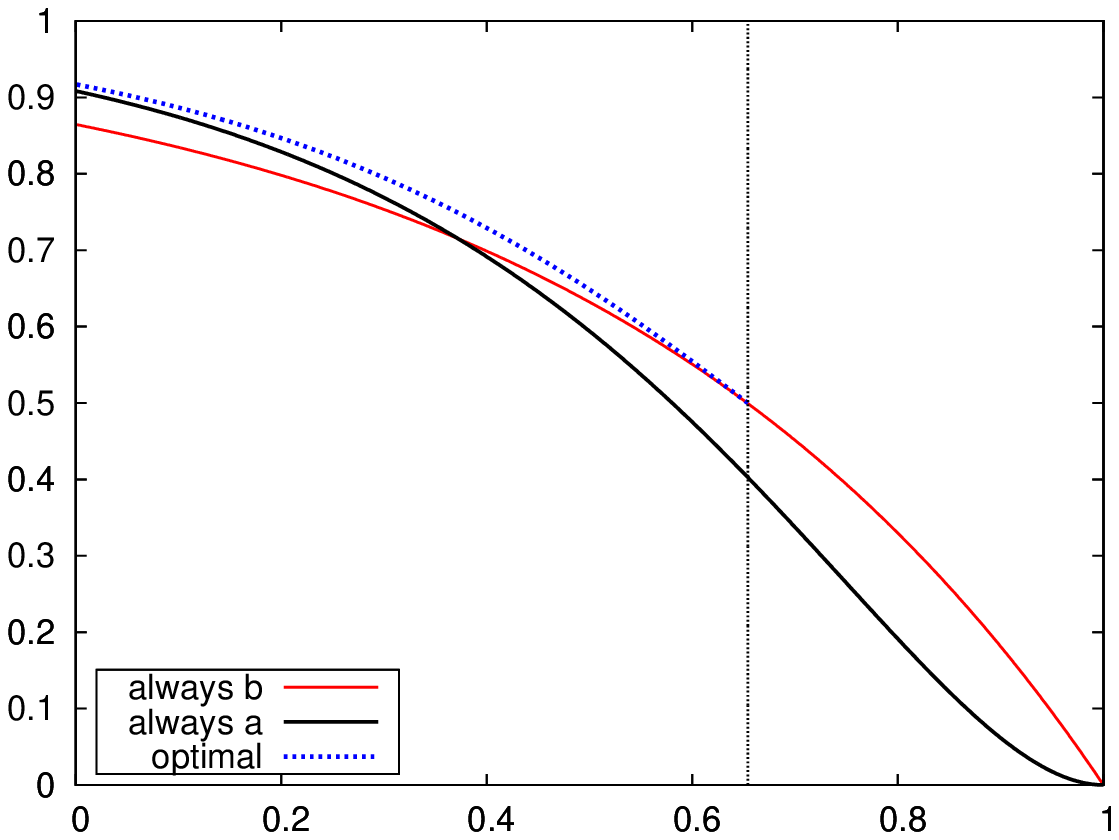}
\caption{A CTMDP and the reachability probabilities for all positional schedulers with time bound $t_{\max}=1$. 
Time $t'=t_{\max}-\frac{1}{2}\log(2)$ is the optimal time to switch to action $b$. \vspace{-.4cm}
}%
\label{fig:automatonAndReachability}%
\end{figure}

Continuous-time Markov decision processes (CTMDPs) are a widely used framework for dependability analysis and for modelling the control of manufacturing processes~\mbox{\cite{Puterman1994,Feinberg/04/MDP}}, because they combine real-time aspects with probabilistic behaviour and non-deterministic choices. 
%
CTMDPs can also be viewed as a framework that unifies different stochastic model types~\cite{sanders1989rbm,Puterman1994,GeneralizedStochasticPetriNetsAjmone95,hermanns2002imc,NeuhausserDelayedNondeterminism09}.

While CTMDPs allow for analysing worst-case and best-case scenarios, they fall short of the demands that arise in many real control problems, as they disregard the different nature that non-determinism can have depending on its source: 
Some sources of non-determinism are supportive, while others are hostile, and in a realistic control scenario, we face both types of non-determinism at the same time:
Supportive non-determinism can be used to model the influence of a controller on the evolution of a system, while hostile non-determinism can capture abstraction or unknown environments. 
We therefore consider a natural extension of CTMDPs: Continuous-time Markov games (CTMGs) that have two players with opposing objectives~\cite{Brazdil+al/09/TimeBoundedReachability}. 

The analysis of CTMDPs and CTMGs requires to resolve the non-deterministic choices by means of a scheduler (which consists of a pair of strategies in the case of CTMGs), and typically tries to optimise a given objective function. 

In this paper, we study the \emph{time-bounded reachability problem}, which recently enjoyed much attention~\cite{
Baier+all/05/efficientCTMDP,ZHHW08,NeuhausserDelayedNondeterminism09,neuhausserzhangTimeBoundedReachabilityReport,Brazdil+al/09/TimeBoundedReachability}.
Time-bounded reachability in CTMDPs is the standard control problem to construct a scheduler that controls the Markov decision process such that the probability of reaching a goal region within a given time bound is maximised (or minimised), and to determine the value.
For CTMGs, time-bounded reachability reduces to finding a Nash equilibrium, that is, a pair of strategies for the 
players, such that each strategy is optimal for the chosen strategy of her opponent.
\pagebreak[3]

While continuous-time Markov games are a young field of research \cite{Brazdil+al/09/TimeBoundedReachability,Rabe+Schewe/10/OptimalSchedulersCTMDP}, Markov decision processes have been studied for decades \cite{Kakumanu/71/CTMDPs,Bellman/57/DP}.

Optimal control in CTMDPs clearly depends on the observational power we allow our schedulers to have when observing a run.
In the literature, various classes of schedulers with different restrictions on what they can observe~\cite{WolovickJohr06Meaningful,NeuhausserDelayedNondeterminism09,Brazdil+al/09/TimeBoundedReachability,Rabe+Schewe/10/OptimalSchedulersCTMDP} are considered.
We focus on the most general class of schedulers, schedulers that can fully observe the system state and may change their decisions at any point in time (\emph{late schedulers}, cf.~\cite{Bellman/57/DP,NeuhausserDelayedNondeterminism09}).
To be able to translate our results to the more widespread class of schedulers that fix their decisions when entering a location (early schedulers), we introduce discrete locations that allow for a translation from early to late schedulers (see Appendix~\ref{app:late2early}). 

Due to their practical importance, time-bounded reachability for continuous-time Markov models has been studied intensively~\cite{Brazdil+al/09/TimeBoundedReachability,Bellman/57/DP,neuhausserzhangTimeBoundedReachabilityReport,Baier+all/05/efficientCTMDP,NeuhausserDelayedNondeterminism09,Aziz/00/exactModelCheckingCTMC,Baier+Katoen+Hermanns/99/approximate,ZHHW08}.
However, most previous research focussed on \emph{approximating} optimal control.
(The existence of optimal control is currently only known for the artificial class of \emph{time-abstract} schedulers~\cite{Brazdil+al/09/TimeBoundedReachability,Rabe+Schewe/10/OptimalSchedulersCTMDP}, which assume that the scheduler has no access whatsoever to a clock.)
While an efficient approximation is of interest to a practitioner, being unable to determine whether or not optimal control \emph{exists} is very dissatisfying from a scientific point of view.

\paragraph{\bf Contributions. }
This paper has three main contributions:
First, we extend the common model of CTMDPs by adding discrete locations, which are passed in $0$ time.
This generalisation of the model is mainly motivated by avoiding the discussion about the appropriate scheduler class.
In particular, the widespread class of schedulers that fix their actions when entering a location 
can be encoded by a simple mapping. 

The second contribution of this paper is the answer to an intriguing research question that remained unresolved for half a century: We show that optimal control of CTMDPs exists for time-bounded reachability and safety objectives.
Moreover, we show that optimal control can always be \emph{finite}.

Our third contribution is to lift these results to continuous-time Markov games.

Pursuing a different research question, we exploit proof techniques that differ from those frequently used in the analysis of CTMDPs.
Our proofs build mainly on topological arguments: The proof that demonstrates the existence of measurable optimal schedulers, for example, shows that we can fix the decisions of an optimal scheduler successively on closures of open sets (yielding only measurable sets), and the lift to finiteness uses local optimality  of positional schedulers in open left and right environments of arbitrary points of times and the compactness of the bounded time interval.

\paragraph{\bf Structure of the Paper. } We follow a slightly unorthodox order of proofs for a mathematical paper: we start with a special case in Section~\ref{sect:TT} and generalise the results later. 
Besides keeping the proofs simple, this approach is chosen because the simplest case, CTMDPs, is the classical case, and we assume that a wider audience is interested in results for these structures. 
In the following section, we strengthen this result by demonstrating that optimal control does not only exist, but can be found among schedulers with finitely many switching points and positional strategies between them.
In Section~\ref{sec:discrete}, we lift this result to single player games (thus extending it to other scheduler classes like those which fix their decision when entering a location, cf.\ Appendix~\ref{app:late2early}).
In the final section, we generalise the existence theorem for finite optimal control to finite co-optimal strategies for general continuous-time Markov games.

\section{Preliminaries}\label{sect:preliminaries}


{A continuous-time Markov game is a tuple $(\locations,\locations_d,\locations_c,\locations_r,\locations_s,G,\act,\ratematrix,\probabilitymatrix,\nu)$, consisting of}
\begin{itemize}
 \item a finite set $\locations$ of locations, which is partitioned into
\begin{itemize}
\item a set $\locations_d$ of \emph{discrete} locations and a set $\locations_c$ of \emph{continuous} locations, and
\item sets $\locations_r$ and $\locations_s$ of locations owned by a \emph{reachability} and a \emph{safety} player, 
\end{itemize}
\item a dedicated set $G \subseteq \locations$ of \emph{goal locations},
\item a finite set $\act$ of actions,
\item a rate matrix $\ratematrix: (\locations_c\times\act\times\locations) \to \Q_{\geqslant0}$,
\item a discrete transition matrix $\probabilitymatrix: (\locations\times\act\times\locations) \to \Q_{\geqslant0}\cap [0,1]$, and
\item an initial distribution $\nu\in\dist(\locations)$,
\end{itemize}
that satisfies the following side-conditions:
For all continuous locations $l\in\locations_c$, there must be an action $a\in\act$ such that $\ratematrix(l,a,\locations):=\sum_{l'\in \locations} \ratematrix(l,a,l')>0$;
we call such actions \emph{enabled}.
For actions enabled in continuous locations, we require $\probabilitymatrix(l,a,l')=\frac{\ratematrix(l,a,l')}{\ratematrix(l,a,\locations)}$, and we require $\probabilitymatrix(l,a,l')=0$ for the remaining actions.
For discrete locations, we require that either $\probabilitymatrix(l,a,l')=0$ holds for all $l' \in \locations$, or that $\sum_{l'\in \locations} \probabilitymatrix(l,a,l')=1$ holds true.
Like in the continuous case, we call the latter actions \emph{enabled} and require the existence of at least one enabled action for each discrete location $l\in \locations_d$. 

The idea behind discrete-time locations is that they execute immediately.
We therefore do not permit cycles of only discrete-time locations (counting every positive rate of any action as a transition).
This restriction is stronger than it needs to be, but it simplifies our proofs, and the simpler model is sufficient for our means.

We assume that the goal region is absorbing, that is $\probabilitymatrix(l,a,l')=0$ holds for all $l \in G$ and $l' \notin G$. 
See Section~\ref{sec:variances} for the extension to non-absorbing goal regions. 

Intuitively, it is the objective of the reachability player to maximise the probability to reach the goal region in a predefined time $t_0$, while it is the objective of the safety player to minimise this probability.
(Hence, it is a zero-sum game.)

We are particularly interested in (traditional) CTMDPs. They are single player CTMGs, where either all positions belong to the reachability player ($\locations=\locations_r$), or to the safety player  ($\locations=\locations_s$), without discrete locations ($\locations_d=\emptyset$ and $\locations_c=\locations$). 


\paragraph{\bf Paths. } A \emph{timed path} $\pi$ in a CTMG $\M$ is a finite sequence in $\locations\times(\act\times\R_{\geqslant0}\times\locations)^*=\paths(\M)$.
We write
\[
  l_0\xrightarrow{a_0,t_0} l_1\xrightarrow{a_1,t_1}
    \cdots ~ \xrightarrow{a_{n-1},t_{n-1}}l_n
\]
for a sequence $\pi,$ and we require $0 \leq t_{i-1} \leq t_{i} \leq t_{\max}$ for all $i<n$, where $t_{\max}$ is the time bound for our time-bounded reachability probability.
(We are not interested in the behaviour of the system after $t_{\max}$.)
The $t_i$ denote the system's time when the action $a_i$ is selected and a discrete transition from $l_i$ to $l_{i+1}$ takes place.
Concatenation of paths $\pi,\pi'$ will be written as $\pi\circ\pi'$ if the last location of $\pi$ is the first location of $\pi'$ and the points of time are ordered correctly. We call a timed path a \emph{complete} timed path when we want to stress that this path describes a complete system run, not to be extended by further transitions.

\paragraph{\bf Schedulers and Strategies.}

The nondeterminism in the system needs to be resolved by a scheduler which maps paths to decisions. 
The power of schedulers is determined by their ability to observe and distinguish paths, and thus by their domain. 
%
In this paper, we consider the following common scheduler classes: 
\begin{itemize}
\item \emph{Timed history-dependent} (TH) schedulers \hfill
	   $\paths(\M)\times\R_{\geqslant0}\rightarrow D$ \hspace{1cm}\mbox{} \\
	   that map timed paths and the remaining time to decisions.
	\item \emph{Timed positional} (TP) schedulers \hfill
	   $\locations\times\R_{\geqslant0}\rightarrow D$ \hspace{1cm}\mbox{} \\
	   that map locations and the remaining time to decisions.
	\item \emph{Positional} (P) or memoryless schedulers \hfill
	   $\locations\rightarrow D$ \hspace{1cm}\mbox{}\\
	   that map locations to decisions. 
\end{itemize}
Decisions $D$ are either randomised (R), in which case $D = \dist(\act)$ is the set of distributions over enabled actions, or are restricted to deterministic (D) choices, that is $D = \act$.
Where it is necessary to distinguish randomised and deterministic versions we will add a postfix to the scheduler class, for example THD and THR. 

\emph{Strategies. } In case of CTMGs, a scheduler consists of the two participating players' strategies, which can be seen as functions $\paths(\M_p)\times\R_{\geqslant0}\rightarrow D$,
where $\paths(\M_p)$ denotes, for $p\in \{r,s\}$, the paths ending on the position of the reachability or safety player, respectively.
As for general schedulers, we can introduce restrictions on what players are able to observe.

\emph{Discrete locations. } The main motivation to introduce discrete locations was to avoid the discussion whether a scheduler has to fix his decision, as to which action it chooses, upon entering a location, or whether such a decision can be revoked while staying in the location. 
For example, the general measurable schedulers discussed in~\cite{WolovickJohr06Meaningful} have only indirect access to the remaining time (through the timed path), and therefore have to decide upon entrance of a location which action they want to perform. 
Our definition builds on fully-timed schedulers (cf.~\cite{Bellman/57/DP}) that were recently rediscovered and formalised by Neuh\"au{\ss}er et al.~\cite{NeuhausserDelayedNondeterminism09}, which may revoke their decision after they enter a location.
(As a side result, we lift Neuh\"au{\ss}er's restriction to local uniformity.)
The discrete locations now allow to encode making the decision upon entering a continuous location $l$ by mapping the decision to a discrete location that is `guarding the entry' to a family of continuous locations, one for each action enabled in~$l$.
(See Appendix~\ref{app:late2early} for details.)

\paragraph{\bf Cylindrical Schedulers. }
While it is common to refer to TH schedulers as a class, the truth is that there is no straightforward way to define a measure for the time-bounded reachability probability for the complete class (cf.\ \cite{WolovickJohr06Meaningful}).
We therefore turn to a natural subset that can be used as a building block for a powerful yet measurable sub-class of TH schedulers, which is based on cylindrical abstractions of paths.

Let $\I$ be a finite partition of the interval $[0,t_{\max}]$ into intervals $I_0=[0,t_0]$ and $I_i=(t_{i-1},t_i]$ for $i=1,\ldots,n$ with $t_0\geq 0$ and $t_i>t_{i-1}$ for $i=1,\ldots,n$, where $t_n = t_{\max}$ is the time-bound from the problem definition.
Then we denote with $[t]_\I$ the interval $I_i \in \I$ that contains $t$, called the \emph{\I-cylindrification of $t$}, and we denote with
$[\pi]_\I = l_0\xrightarrow{a_0,[t_0']_\I} l_1\xrightarrow{a_1,[t_1']_\I} \cdots ~ \xrightarrow{a_{n-1},[t_{n-1}']_\I}l_n$ the \emph{$\I$-cylindrification} of the timed path
$\pi=l_0\xrightarrow{a_0,t_0'} l_1\xrightarrow{a_1,t_1'} \cdots ~ \xrightarrow{a_{n-1},t_{n-1}'}l_n$.

We call a TH scheduler \emph{$\I$-cylindrical} if its decisions depend only on the \emph{cylindrification} $[\pi]_\I$ and $[t]_\I$ of $\pi$ and $t$, respectively, and \emph{cylindrical} if it is $\I$-cylindrical for some finite partition $\I$ of the interval $I=[0,t_{\max}]$.

\paragraph{\bf Cylindrical Sets and Probability Space. }
For a given finite partition $\I$ of the interval $[0,t_{\max}]$, an \emph{$\I$-cylindrical} set of timed paths is the set of timed paths with the same $\I$-cylindrification, and
we call a finite partition $\I'$ of $[0,t_{\max}]$ a \emph{refinement} of $\I$ if every interval in $\I$ is the union of intervals in $\I'$.

For an $\I$-cylindrical scheduler $\S$ and an $\I'$-cylindrical set of finite timed paths, where $\I'$ is a refinement%
\footnote{The restriction to partitions $\I'$ that refine $\I$ is purely technical, because for arbitrary $\I'$ we can simply use a partition $\I''$ that refines both $\I$ and $\I'$, and reconstruct every $\I'$-cylindrical set as a finite union of $\I''$-cylindrical sets.}
of $\I$, the likelihood that a complete path is from this cylindrical set is easy to define:
Within each interval of $\I$, the likelihood that a CTMDP $\M$ with scheduler $\S$ behaves in accordance with the $\I'$-cylindrical set can---assuming compliance in all previous intervals---be checked like for a finite Markov chain.

The probability $p_{I_i}$ to comply with the $i$-th segment of the partition $\I'$ of $[0,t_{\max}]$ is the product of three multiplicands $(p_{I_i}=p_1^{I_i} \cdot p_2^{I_i} \cdot p_3^{I_i})$:
\begin{enumerate}
\item the probability $p_1^{I_i}$ that the actions are \emph{chosen} in accordance with the $\I'$-cylindrical set of timed paths (which is either $0$ or $1$ for discrete schedulers, and the product of the likelihood of the individual decisions for randomised schedulers),

\item the probability $p_2^{I_i}$ that the transitions are \emph{taken} in accordance with the $\I'$-cylindrical set of timed paths, provided the respective actions are chosen, which is simply the product over the individual probabilities $\probabilitymatrix(l_i,a_i,l_{i+1})$ in this sequence of the $\I'$-cylindrical set of timed paths, and
\item the probability $p_3^{I_i}$ that the \emph{right number of steps} is made in this sequence of the $\I'$-cylindrical set of timed paths.
\end{enumerate}

The latter probability $p_3^{I_i}$ is $0$ if the last location is a discrete location, as the system would leave this location at the same point in time in which it was entered.
Otherwise, it is the difference $p_3^{I_i}=p_4^{I_i}-p_5^{I_i}$ between the likelihood that at least the correct number of $n \geq 0$ transitions starting in continuous locations are made ($p_4^{I_i}$), and the likelihood that at least $n+1$ transitions starting in continuous locations are made ($p_5^{I_i}$) in the relevant sequence of the timed path.

Let $\overline{l}_0,\overline{l}_1,\ldots,\overline{l}_n$ be the $n$ continuous locations (named in the required order of appearance; note that $n$ might be $0$, and that the same location can occur multiple times), and let $\lambda_0,\lambda_1,\ldots,\lambda_n$ be the transition rate one would observe at the respective $\overline{l}_i$. For deterministic schedulers, this transition rate is simply $\lambda_i=\ratematrix(\overline{l}_i,a_i,\locations)$, where $a_i$ is the decision from $\S$ at the respective position in a timed path and in $I_i$. For a randomised scheduler $\S$, it is the respective expected transition rate $\lambda_i=\sum_{a\in \act(\overline{l}_i)} p_a \ratematrix(\overline{l}_i,a,\locations)$, where $p_a$ is the likelihood that $\S$ makes the decision $a$ at the respective position in a timed path and in $I_i$.
Note that the locations and transition rates are fixed.

The likelihood to get a path of length $\geq n$ is then \\ $\int_{(\tau_0,\ldots,\tau_{n-1})\in \Phi_{n,i}} \prod_{k=0}^{n-1} \lambda_k e^{-\lambda_k \tau_k}d\tau_k$ for $\Phi_{n,i} = \{(\tau_0,\ldots,\tau_{n-1})\in [0,t_{\max}]^n \mid \\ \sum_{j=0}^{n-1} \tau_j \leq t_i - t_{i-1}\}$ for $n>0$, and $1$ for $n=0$.
Likewise, the likelihood to get a path of length  $\geq n+1$ is $\int_{(\tau_0,\ldots,\tau_{n})\in \Phi_{n+1}^i} \prod_{k=0}^{n}\lambda_k e^{-\lambda_k \tau_k}d\tau_k$ for $\Phi_{n+1,i} = \{(\tau_0,\ldots,\tau_{n})\in [0,t_{\max}]^{n+1} \mid \sum_{j=0}^{n} \tau_j \leq t_i - t_{i-1}\}$.
(Recall that $t_i$ and $t_{i-1}$ are the upper and lower endpoints of the interval $I_i$.)

The likelihood that a complete timed path is in the $\I'$-cylindrical set of timed paths for $\S$ is the product $\prod_{I \in \I'}p_I$ over the individual $p_{I_i}$.

\paragraph{\bf Probability Space.}
Having defined a measure for the likelihood that, for a given cylindrical scheduler, a complete timed path is in a particular cylindrical set, we define the likelihood that it is in a finite union of disjoint sets of cylindrical paths as the sum over the likelihood for the individual cylindrical sets.

This primitive probability measure for primitive schedulers can be lifted in two steps by a standard space completion, going from this primitive measures to quotient classes of Cauchy sequences of such measures:
\begin{enumerate}
\item In a first step, we complete the space of measurable sets of complete timed paths from finite unions of cylindrical sets of timed paths to Cauchy sequences of finite unions of cylindrical sets of timed paths.
We define the required difference measure of two sets of timed paths (each a finite disjoint union of cylindrical sets) as the measure of the symmetrical difference of the two sets.
This set can obviously be represented as a finite disjoint union of cylindrical sets of timed paths, and we can use our primitive measure to define this difference measure.

\item Having lifted the measure to this completed space of paths, we lift the set of measurable \emph{schedulers} in a second step from cylindrical schedulers to Cauchy sequences of cylindrical schedulers.
(The difference measure between two cylindrical schedulers is the likelihood that two schedulers act observably different.)
\end{enumerate}

More details of these standard constructions can be found in Appendix~\ref{app:ourmeasure}.

\paragraph{\bf Time-Bounded Reachability Probability. }
For a given CTMG $(\locations,\locations_d,\locations_c,\locations_r,\locations_s,G,\act,\ratematrix,\probabilitymatrix,\nu)$ and a given measurable scheduler $\S$ that resolves the non-determinism, we use the following notations for the probabilities:
\begin{itemize}
\item $\prob_{\S}^{\M}(l,t)$ is the probability of reaching the goal region $G$ within time $t$ when starting in location $l$,

\item $\prob_{\S}^{\M}(t) = \sum_{l \in \locations}\nu(l)\prob_{\S}^{\M}(l,t)$ denotes the probability of reaching the goal region $G$ within time $t$.


\end{itemize}

As usual, the supremum 
of the time-bounded reachability probability over a particular scheduler class is called the time-bounded reachability of $\M$ for this scheduler class. 

\section{Optimal Scheduling in CTMDPs}\label{sect:TT}

In this section, we demonstrate the existence of optimal schedulers in traditional CTMDPs.
%
%
%
Before turning to the proof, let us first consider what happens if time runs out, that is, at time $t_{\max}$, and then develop an intuition what an optimal scheduling policy should look like.

If we are still in time ($t \leq t_{\max}$) and we are in a goal location, then we reach a goal location in time with probability $1$; and
if time has run out ($t=t_{\max}$) and we are not in a goal location, then we reach a goal location in time with probability $0$.
For ease of notation, we also fix the probability of reaching the goal location in time to $0$ for all points in time strictly \emph{after} $t_{\max}$.
For a measurable TPR scheduler $\S$, we would get:
\begin{itemize}
 \item $\prob_{\S}^{\M}(l,t) = 1$ holds for all goal locations $l \in G$ and all $t \leq t_{\max}$,
 \item $\prob_{\S}^{\M}(l,t_{\max}) = 0$ holds for all non-goal locations $l \notin G$, and
 \item $\prob_{\S}^{\M}(l,t) = 0$ holds for all locations $l \in \locations$ and all $t > t_{\max}$.
\end{itemize}

A scheduler $\S$ can, in every point in time, choose from distributions over successor locations.
Such a choice should be optimal, if the expected gain in the probability of reaching the goal location is maximised.

This gain has two aspects: first, the probability of reaching the goal location \emph{provided a transition is taken}, and second, the likelihood of \emph{taking a transition}.
Both are multiplicands in the defining differential equations, assuming a cylindrical TPD scheduler $\S$
\[-\dot{\prob}{}_{\S}^{\M}(l,t) = \sum_{l'\in\locations} \ratematrix\big(l,\S(l,t),l'\big) \cdot \left(\prob_{\S}^{\M}(l',t) -\prob_{\S}^{\M}(l,t)\right).\]
(We skip the simple generalisation to measurable TPD scheduler because it is not required in the following proofs.)

The reachability probability of any scheduler is therefore intuitively dominated by the function $f_{\max}$ and dominates the function $f_{\min}$ defined by the following equations:
\[-\dot{f}_{\mathsf{opt}}(l,t) = \hspace*{-3mm}
{\begin{array}{c}
\vspace*{-4pt} \\
\mathsf{opt} \vspace*{-4pt}\\ \scriptsize \mbox{$a\in\act(l)$}
\end{array}}\sum_{l'\in\locations} \ratematrix(l,a,l') \cdot \left(f_{\mathsf{opt}}(l',t) -f_{\mathsf{opt}}(l,t)\right)\hfill \mbox{for }t\in [0,t_{\max}],\]
where $\mathsf{opt}\in \{\min,\max\}$. This intuitive result is not hard to prove%
\footnote{The systems of non-linear ordinary differential equations used in this paper are all quite obvious, and the challenge is to prove that they can be \emph{taken} and not merely approximated.
An approximative argument for these ODE's goes back to Bellman \cite{Bellman/57/DP}, but he uses a less powerful set of schedulers, and only proves that $f_{\max}$ and $f_{\min}$ can be approximated from below and above, respectively, claiming that the other direction is obvious.
After starting with a similar claim, we were urged to include a full proof.}.

\begin{lemma}\label{lem:stupid}
The reachability probability of any measurable THR scheduler is dominated by the function $f_{\max}$ and dominates the function $f_{\min}$. 
\end{lemma}

\begin{proofidea}
To proof this claim for $f_{\max}$, assume that there is a scheduler that provides a better time-bounded reachability probability $\prob_{\S}^{\M}(l,t)>f_{\max}(l,t)$ for some location $l\in \locations$ and time $t\in [0,t_{\max}]$ (in particular for $t=0$), and hence improves over $f_{\max}(l,t)$ at this position by at least $3\varepsilon$ for some $\varepsilon >0$.

$\S$ is a Cauchy sequence of cylindrical schedulers. Therefore we can sacrifice one $\varepsilon$ and get an $\varepsilon$-close cylindrical scheduler from this sequence, which is still at least $2\varepsilon$ better than $f_{\max}$ at position $(l,t)$.

As the measure for this cylindrical scheduler is a Cauchy sequence of measures for sequences with a bounded number of discrete transitions, we can sacrifice another $\varepsilon$ to sharpen the requirement for the scheduler to reach the goal region in time \emph{and} with at most $n_\varepsilon$ steps for an appropriate bound $n_\varepsilon \in \mathbb N$, still maintaining an $\varepsilon$ advantage over $f_{\max}$. Hence, we can compare with a finite structure, and use an inductive argument to show for paths $\pi$ of shrinking length that end in any location $l'\in \locations$ that $f_{\max}(l',t) \leq \prob_{\S}^{\M}(\pi,t)$ holds true.
\end{proofidea}

The full proof is moved to Appendix~\ref{app:proof}.

%
%

\begin{theorem}\label{thm:pureexistence}
For a CTMDP, there is a measurable TPD scheduler $\S$ optimal for maximum time-bounded reachability in the class of measurable THR scheduler.
\end{theorem}

\begin{proof}
We construct a \emph{measurable} scheduler $S$ that always chooses an action $a$ that maximises $\sum_{l'\in\locations} \ratematrix\big(l,\S(l,t),l'\big) \cdot \left(\prob_{\S}^{\M}(l',t) -\prob_{\S}^{\M}(l,t)\right)$;
by Lemma~\ref{lem:stupid}, this guarantees that $\sum_{l \in \locations} \nu(l) Pr_{\S}^{\M}(l,0) = \sum_{l \in \locations} \nu(l) f_{\max}(l,0) = \sup\limits_{\S \in THR} Pr_{\S}^{\M}(t_{\max})$ holds true.

To construct the scheduler decisions for a location $l$ for a measurable scheduler $\S$, we partition $[0,t_{\max}]$ into measurable sets $\{D_a\mid a \in \act(l)\}$, such that $\S$ only makes decisions that maximise $\sum_{l'\in\locations}{\ratematrix(l,a,l') \cdot \big(\prob_{\S}^{\M}(l',t) - \prob_{\S}^{\M}(l,t)\big)}$.
(For positions outside of $[0,t_{\max}]$, 
the behaviour of the scheduler does not matter. $\S(l,t)$ can therefore be fixed to any constant decision $a \in \act(l)$ for all $l\in \locations$ and $t\notin [0,t_{\max}]$.)

We start with fixing an arbitrary order $\succ$ on the actions in $\act(l)$ and introduce, for each point $t \in [0,t_{\max}]$, an order $\curlyeqsucc_t$ on the actions determined by the value of $\sum_{l'\in\locations}\ratematrix(l,a,l') \cdot \big(f_{\max}(l',t) -f_{\max}(l,t)\big)$, using $\succ$ as a tie-breaker.

Along the order of $\succ$, we construct, starting with the minimal element, for each action $a\in \act(l)$:
\begin{enumerate}
 \item Open sets $O_a$ that contain the positions where our scheduler does \emph{not} make a decision%
\footnote{Note that, for all $a' \prec a$, the points $T_{a'}$ in time where the scheduler does make the decision $a'$ have been fixed earlier by this construction.}
$a' \prec a$.

(We choose $O_a=[0,t_{\max}]$ for the action $a$ that is minimal with respect to $\succ$.)

\item A set $T_a$ that is open in $O_a$ and contains the points in time in $O_a$ where $a$ is maximal with respect to $\curlyeqsucc_t$.
\item A set $D_a = \overline{T_a} \cap O_a$ which is the closure of $T_a$ in $O_a$.

If $a$ is not maximal, we set $O_{a'}=O_a \smallsetminus D_a$ for the successor $a'$ of $a$ with respect to~$\succ$.
\end{enumerate}

To complete the proof, we have to show that the scheduler $\S$, which chooses $a$ for all $t\in D_a$, makes only decisions that maximise the gain, that is $\sum_{l'\in\locations}\ratematrix(l,a,l') \cdot $ $\big(f_{\max}(l',t) - f_{\max}(l,t)\big)$ (and hence that $f_{\max}=\prob_{\S}^{\M}$ holds true), and we have to show that the resulting scheduler is measurable.
As an important lemma on the way, we have to demonstrate the claimed openness of the $O_a$'s and $T_a$'s in the compact Euclidean space $[0,t_{\max}]$.

This openness is provided by a simple inductive argument:
First, the complete space $[0,t_{\max}]$ is open in itself.

Let us assume that $a$ is maximal w.r.t.\ $\curlyeqsucc_t$ for a $t$ in the open set $O_a$.
Then the following holds:  $\sum_{l'\in\locations}{\ratematrix(l,a,l') \cdot \big(f_{\max}(l',t) -f_{\max}(l,t)\big)}$ is \emph{strictly} greater for $a$ compared to the respective value of all other actions $a' \succ a$ (because $\succ$ serves as tie-breaker), and hence this holds for some $\varepsilon$-environment of $t$ that is contained in the open set $O_a$.
For the actions $a' \prec a$ in this $\varepsilon$-environment of $t$, the respective value also cannot be \emph{strictly} greater compared to $a$, because otherwise one of these actions had been selected before.

Note that this argument provides optimality of the choices in $T_a$ as well as openness of $T_a$.
The optimality for the choices at the fringe of $T_a$ (and hence the extension of the optimality argument to $D_a$) is a consequence of the continuity of $f_{\max}$.

As every open and closed set in $\mathbb R$ is (Lebesgue) measurable, $O_a$ (or, to be precise, $O_a\cap (0,t_{\max})$) and $\overline{T_a}$, and hence their intersection $D_a = \overline{T_a} \cap O_a$, are measurable.

Our construction therefore provides us with a \emph{measurable} scheduler, which is optimal, deterministic, and timed positional.
\qed
\end{proof}

%
%
%
%
%

By simply replacing maximisation by minimisation, $\sup$ by $\inf$, and $\max$ by $\min$, we can rewrite the proof to yield a similar theorem for the minimisation of time-bounded reachability, or likewise, for the maximisation of time-bounded safety.

\begin{theorem}\label{thm:puresafety}
For a CTMDP, there is a measurable TPD scheduler $\S$ optimal for minimum time-bounded reachability in the class of measurable THR scheduler.
\end{theorem}

\section{Finite Optimal Control}\label{sect:finite}

In this section we show that, once the existence of an optimal scheduler is established, we can refine this result to the existence of a \emph{cylindrical} optimal TPD scheduler, that is, a scheduler that changes only finitely many times between different positional strategies.
This is as close as we can hope to get to implementability as optimal points for policy switching are---like in the example from Figure \ref{fig:automatonAndReachability}---almost inevitably irrational.

Our proof of Theorem~\ref{thm:pureexistence} makes a purely topological existence claim, and therefore does not imply that a finite number of switching points suffices.
In principle, this could mean that the required switching points have one or more limit points, and an unbounded number of switches is required to optimise time-bounded reachability.
$x\cdot sin (x^{-1})$ (cf.\ Figure~\ref{fig:limitpoint}) is an example for a continuous function for which the codomain of $0$ has a limit point at $0$, and the right curve of Figure~\ref{fig:limitpoint} shows the derivations for a positional scheduler (black) and a potential comparison with a gain function such that their intersections have a limit point.

\begin{figure}[t]
\includegraphics[width=0.48\columnwidth]{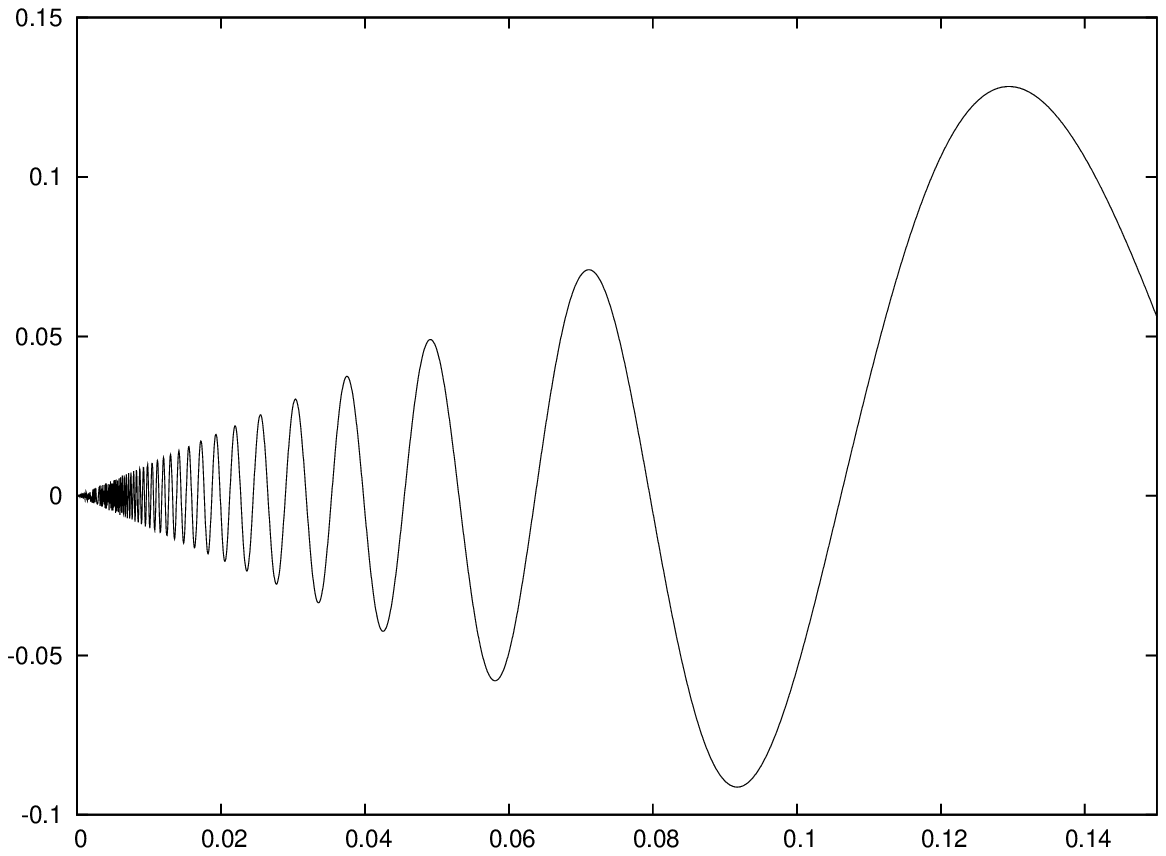}
\includegraphics[width=0.48\columnwidth]{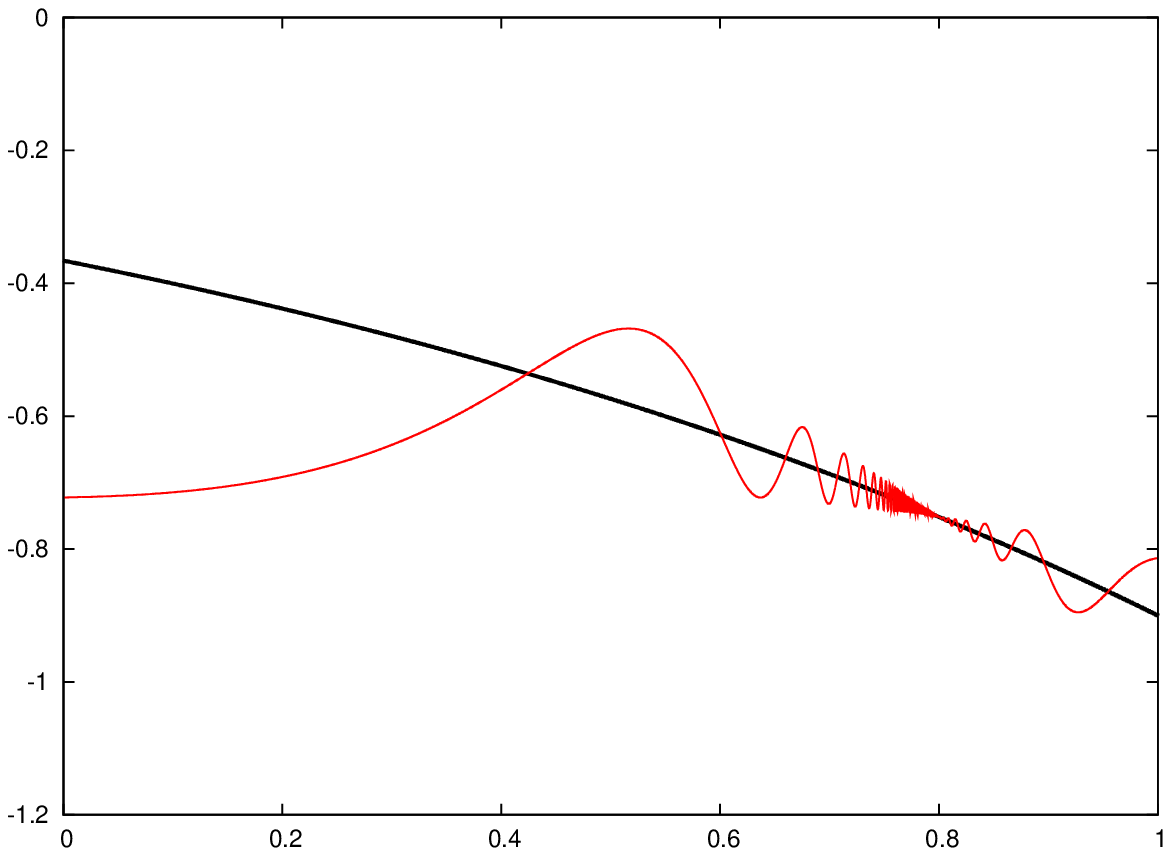}\hfill
\caption{
The plot on the left shows the function $x\cdot\sin(x^{-1})$ as an example for the limit point problem. 
The plot on the right: Theorem \ref{thm:pureexistence} would not exclude that the intersection points of derivatives (the 'loss') of two positional schedulers have a limit point.
}
\label{fig:limitpoint}%
\end{figure}

To exclude such limit points, and hence to prove the existence of an optimal scheduler with a finite number of switching points, we re-visit the differential equations that define the reachability probability, but this time to answer a different question:
Can we use the true values in some point of time to locally find an optimal strategy for an $\varepsilon$-environment?
If yes, then we could exploit the compactness of $[0,t_{\max}]$:
We could, for all points in time $t\in [0,t_{\max}]$, fix a decision that is optimal in an $\varepsilon$-environment of $t$.
This would provide an open set with a positional optimal strategy around each $t\in [0,t_{\max}]$, and hence an open coverage of a compact set, which would imply a final coverage with segments of positional optimal strategies.

While this is the case for most points, this is not necessarily the case at our switching points.
In the remainder of this section, we therefore show something similar:
For every point $t \in [0,t_{\max}]$ in time, there is a positional strategy that is optimal in a left $\varepsilon$-environment of $t$ (that is, in a set $(t-\varepsilon,t]\cap [0,t_{\max}]$), and one that is optimal in a right $\varepsilon$-environment of $t$.
Hence, we get an open coverage of strategies with at most one switching point, and thus obtain a strategy with a finite number of switching points.

\begin{theorem}
\label{theo:finite}
For every CTMDP, there is a cylindrical TPD scheduler $\S$ optimal for maximum time-bounded reachability in the class of measurable THR scheduler.
\end{theorem}

\begin{proof}
We have seen that the true optimal reachability probability is defined by a system of differential equations.
In this proof we consider the effect of starting with the `correct' values for a time $t\in [0,t_{\max}]$, but
\emph{locally fix a positional strategy} for a small left or right $\varepsilon$-environment of $t$.
That is, we consider only schedulers that keep their decision constant for a (sufficiently) small time $\varepsilon$ before or after~$t$.

Given a CTMDP $\M$, we consider the differential equations that describe the development near the support point $f_{\max}(l,t)$ for each location $l$ under a positional strategy~$D$:
\[
-\dot{\prob}{}^D_l(\tau) = \sum_{l'\in\locations} \ratematrix(l,a_l,l') \cdot \big(\prob_{l'}^D(\tau) - \prob_{l}^D(\tau)\big),\vspace{-1mm}
\]
where $a_l$ is the action chosen at $l$ by $D$ 
(see Figure~\ref{fig:developmentAtSwitchingPoint} for an example).
 
\begin{figure}[t]
\hspace{-.6cm}
\rput(4.6,.8){$t'$}
\includegraphics[width=0.52\columnwidth]{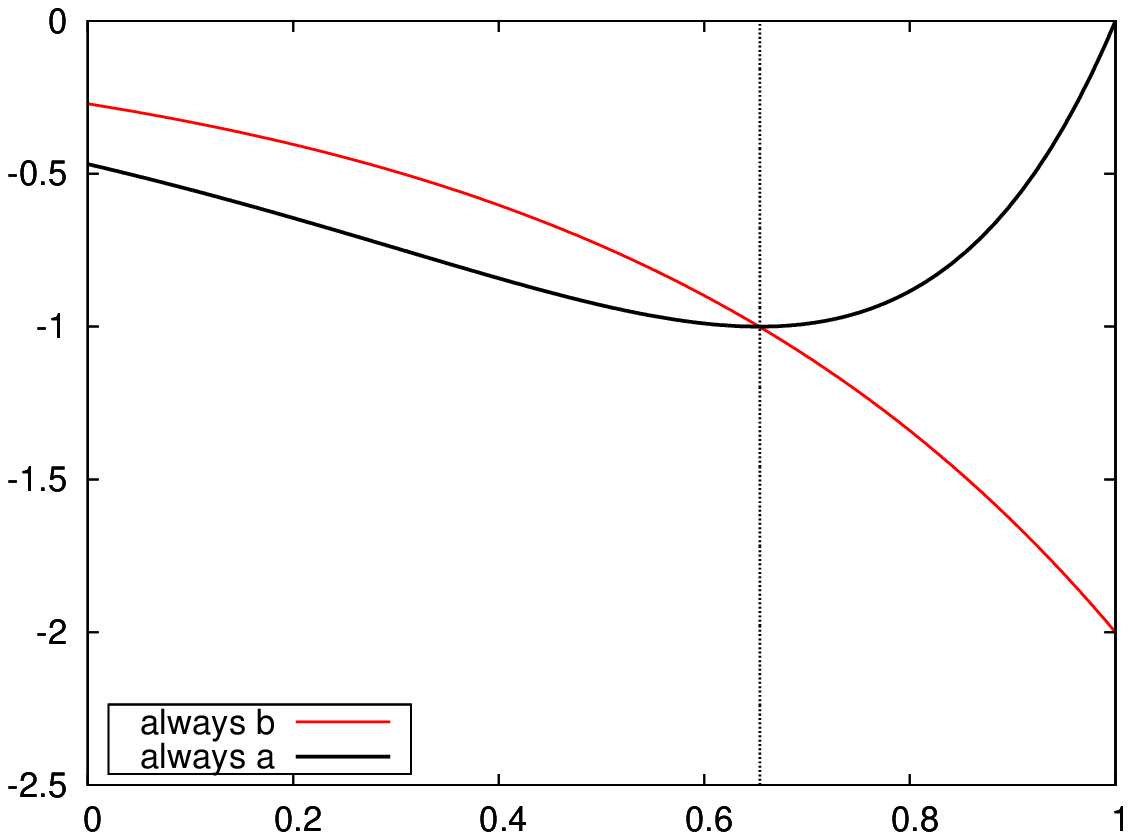}%
\includegraphics[width=0.52\columnwidth]{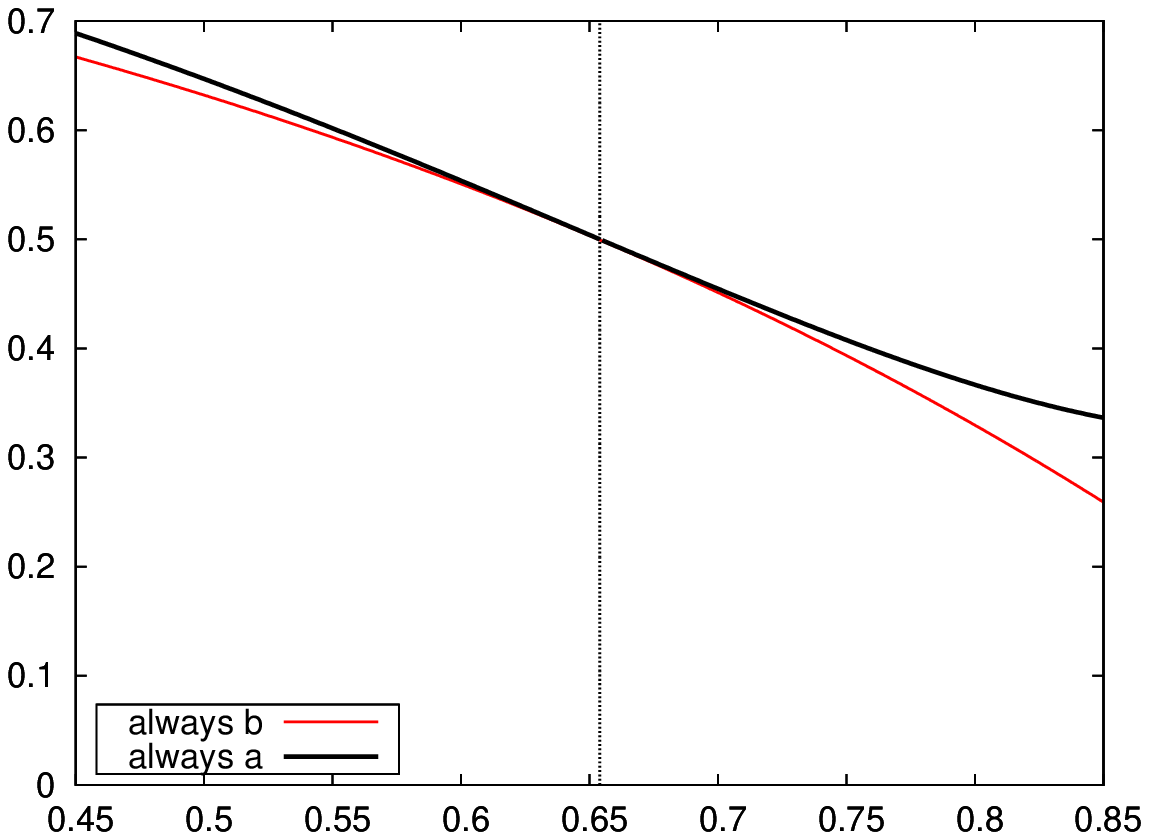}
\caption{The gain functions of the competing stationary strategies of Figure~\ref{fig:automatonAndReachability}. To the left: Developed \emph{gains} $\dot{\prob}_A^{\S_a/\S_b}(t)$ from $t=t_{\max}=1$ in order to find the only switching point $t'=t_{\max}-\frac{1}{2}\log(2)$. To the right: Developed values $\prob_A^{\S_a/\S_b}(t)$ (not gains) from $t'$ in both directions, to show that action $a$ is better for all $t<t'$ whereas $b$ is better for all $t>t'$. This construction is also used to determine the existence of $\epsilon$-environments with a stable strategy around every point~$t$.
}
\label{fig:developmentAtSwitchingPoint}
\end{figure}

Different to the development of the true probability, the development of these linear differential equations provides us with smooth functions.
This provides us with more powerful techniques when comparing two locally positional strategies:
Each deterministic scheduler defines a system $\dot{y} = A y$ of ordinary homogeneous linear differential equations with constant coefficients.

As a result, the solutions $\prob^D_l(\tau)$ of these differential equations---and hence their differences $\prob^{D'}_l(\tau)-\prob^D_l(\tau)$---can be written as finite sums $\sum_{i=1}^n P_i(\tau) e^{\lambda_i \tau}$, where $P_i$ is a polynomial and the $\lambda_i$ may be complex.
Consequently, these functions are holomorphic.

Using the identity theorem for holomorphic functions, $t$ can only be a limit point of the set of $0$ points of $\prob^{D'}_l(\tau)-\prob^D_l(\tau)$ if $\prob^{D'}_l(\tau)$ and $\prob^D_l(\tau)$ are identical on an $\varepsilon$-environment of $t$.
The same applies to their derivations: $\dot{\prob}{}^{D'}_l(\tau)-\dot{\prob}{}^D_l(\tau)$ either has no limit point in $t$, or $\dot{\prob}{}^{D'}_l(\tau)$ and $\dot{\prob}{}^D_l(\tau)$ are identical on an $\varepsilon$-environment of $t$.

For the remainder of the proof, we fix, for a given time $t$, a sufficiently small $\varepsilon>0$ such that, for each pair of schedulers $D$ and $D'$ and every location $l\in\locations$, 
$\dot{\prob}{}^{D'}_l(\tau)-\dot{\prob}{}^{D}_l(\tau)$ is either $<0$, $=0$, or $>0$ on the complete interval $L_\varepsilon^t=(t-\varepsilon,t)\cap [0,t_{\max}] \ni \tau$, and, possibly with different sign, for the complete interval~\mbox{$R_\varepsilon^t=(t,t+\varepsilon)\cap [0,t_{\max}] \ni \tau$}.

We argue the case for the left $\varepsilon$-environment $L_\varepsilon^t$.
In the `$>$' case for a location $l$, we say that $D$ is \emph{$l$-better} than $D'$.
We call $D$ \emph{preferable} over $D'$ if $D'$ is not $l$-better than $D$ for any location $l$, and \emph{better} than $D'$ if $D$ is preferable over $D'$ and $l$-better for some $l\in \locations$.

If $D'$ is $l$-better than $D$ in exactly a non-empty set $\locations_b\subset\locations$ of locations, then we can obviously use $D'$ to construct a strategy $D''$ that is better than $D$ by switching to the strategies of $D'$ in exactly the locations $\locations_b$.

Since we choose our strategies from a finite domain---the deterministic positional schedulers---this can happen only finitely many times.
Hence we can stepwise \emph{strictly} improve a strategy, until we have constructed a strategy $D_{\max}$ preferable over all others.

By the definition of being preferable over all other strategies, $D_{\max}$ satisfies
$$-\dot{\prob}{}^{D_{\max}}_l(\tau) = \max_{a\in\act(l)}\sum_{l'\in\locations} \ratematrix(l,a,l') \cdot \big(\prob_{l'}^{D_{\max}}(\tau) - \prob_{l}^{D_{\max}}(\tau)\big)$$
for all $\tau \in L_\varepsilon^t$ and all $l\in \locations$.

We can use the same method for the right $\varepsilon$-environment $R_\varepsilon^t$, and pick the decision for $t$ arbitrarily;
we use the decision from the respective left $\varepsilon$ environment.

Now we have fixed, for an $\varepsilon$-environment of an arbitrary $t\in [0,t_{\max}]$, an optimal scheduler with at most one switching point.
As this is possible for all points in $[0,t_{\max}]$, the sets $I_\varepsilon^t = L_\varepsilon^t \cup R_\varepsilon^t$ define an open cover of $[0,t_{\max}]$.
Using the compactness of $[0,t_{\max}]$, we infer a finite sub-cover, which establishes the existence of a strategy with a finite number of switching points.
\qed
\end{proof}

The proof for the minimisation of time-bounded reachability (or  maximisation of time-bounded safety) runs accordingly.

\begin{theorem}
\label{theo:finitea}
For every CTMDP, there is a cylindrical TPD scheduler $\S$ optimal for minimal time-bounded reachability in the class of measurable THR scheduler.
\end{theorem}

\section{Discrete Locations}
\label{sec:discrete}

In this section, we treat the mildly more general case of single player CTMGs, which are traditional CTMDPs plus discrete locations.
We reduce the problem of finding optimal measurable schedulers for CTMGs first to \emph{simple CTMGs}, CTMGs whose discrete locations have no incoming transitions from continuous locations.
(They hence can only occur initially at time $0$.)
The extension from CTMDPs to simple CTMGs is trivial.

\begin{lemma}
\label{lem:ssG}
For a simple single player CTMG with only a reachability (or only a safety) player, there is an optimal deterministic scheduler with finitely many switching points.
\end{lemma}

\begin{proof}
By the definition of simple single player games, the likelihood of reaching the goal location from any continuous location and any point in time is independent of the discrete locations and their transitions. For continuous locations, we can therefore simply reuse the results from the Theorems \ref{theo:finite} and \ref{theo:finitea}.

We can only be in discrete locations at time $0$, and for every continuous location $l$ there is a fixed time-bounded reachability probability described by $f_{\mathsf{opt}}(l,0)$.
We can show that there is a timed-positional (even a positional) deterministic optimal choice for the discrete locations at time $t=0$ by induction over the maximal distance to continuous locations:
If all successors have been evaluated, we can fix an optimal timed-positional choice. We can therefore use discrete positions with maximal distance $1$ as induction basis, and then apply an induction step from positions with distance $\leq n$ to positions with distance $n+1$.
\qed
\end{proof}

Rebuilding a single player CTMG $\G$ to a \emph{simple} single player CTMG $\G_s$ can be done in a straight forward manner; it suffices to pool all transitions taken between two continuous locations.
To construct the resulting simple CTMG $\G_s$, we add new continuous locations for each possible time abstract path from continuous locations of the CTMG $\G$, and we add the respective actions:
For continuous locations $l_c,l_c'\in\locations_c$ and discrete locations $l_1^d,\ldots,l_n^d\in \locations_d$ a timed path
$l_c\xrightarrow{a_0,t} l_1^d \xrightarrow{a_1,t} l_2^d \cdots ~ l_n^d \xrightarrow{a_{n},t}l_c'$ translates to
$l_c\xrightarrow{\mathbf a,t} \underline{\xrightarrow{a_0} l_1^d \xrightarrow{a_1} l_2^d \cdots ~ l_n^d \xrightarrow{a_n}l_c'}$,
where the underlined part is a new continuous location.
(For simplicity, we also translate a timed path $l_c\xrightarrow{a,t} l_c'$ to $l_c\xrightarrow{\mathbf a,t}\underline{\xrightarrow{a} l_c'}$.)

The new \emph{actions} of the resulting simple single player CTMG encode the sequences of actions of $\G$ that a scheduler could make in the current location plus in all possible sequences of discrete locations, until the next continuous location is reached.
(Note that this set is finite, and that the scheduler makes all of these transitions at the same point of time.) If $\mathbf a$ encodes choices that depend only on the position (but not on this local history), $\mathbf a$ is called positional.
For continuous locations, all old actions are deleted, and all new continuous locations that end in a location $l_c\in \locations_c$ get the same outgoing transitions as $l_c$.
The rate matrix is chosen accordingly.

Adding the information about the path to locations allows to reconstruct the timed history in the single player CTMG from a history in the constructed simple CTMG.

\begin{theorem}
\label{theo:discrete}
For a single player CTMG $\G$ with only a reachability (or only a safety) player, there is an optimal deterministic scheduler with finitely many switching points.
\end{theorem}

\begin{proof}
First, every scheduler $\S$ for $\G$ can be naturally translated into a scheduler of $\S_s$ of $\G_s$, because every timed-path in $\G_s$ defines a timed-path in $\G$;
the resulting time-bounded reachability probability coincides.

Let us consider a cylindrical optimal deterministic scheduler $\S_{\mathsf{opt}}$ for the simple Markov game, and the function $f_{\mathsf{opt}}$ defined by it.
For the actions $\mathbf a$ $\S_{\mathsf{opt}}$ chooses, we can, for each interval in which $\S_{\mathsf{opt}}$, is positional, use an inductive argument similar to the one from the proof of Lemma~\ref{lem:ssG} to show that we can choose a \emph{positional} action $\mathbf a'$ instead.
The resulting cylindrical deterministic scheduler $\S_{\mathsf{opt}}'$ defines the same $f_{\mathsf{opt}}$ (same differential equations).

Clearly, $f_{\mathsf{opt}}(l_c,t) = f_{\mathsf{opt}}(\underline{\ldots \xrightarrow{a}l_c},t)$ holds true.
We use this observation to change $\S_{\mathsf{opt}}'$ to $\S_{\mathsf{opt}}''$ by choosing the action that $\S_{\mathsf{opt}}'$ chooses for $l_c$ for all locations $\underline{\ldots \xrightarrow{a}l_c}$ and at each point of time.
The resulting scheduler $\S_{\mathsf{opt}}''$ is still cylindrical and deterministic, and defines the same $f_{\mathsf{opt}}$ (same differential equations).

$\S_{\mathsf{opt}}''$ is also the mapping of a cylindrical optimal deterministic scheduler for $\G$.
\qed
\end{proof}

\section{Continuous-Time Markov Games}
\label{sec:games}

In this section, we lift our results from single player to general continuous-time Markov games.
In general continuous-time Markov games, we are faced with two players with opposing objectives:
A reachability player trying to maximise the time-bounded reachability probability, and a safety player trying to minimise it---we consider a $0$-sum game.

Thus, all we need to do for lifting our results to games is to show that the quest for optimal strategies for single player games discussed in the previous section can be generalised to a quest for co-optimal strategies---that is, for Nash equilibria---in general games.
To demonstrate this, it essentially suffices to show that it is not important whether we first fix the strategy for the reachability player and then the one for the safety player in a strategy refinement loop, or vice versa.

Let us first assume CTMGs without discrete locations.

\begin{lemma}
\label{lem:epsilon}
Using the $\varepsilon$-environments $I_\varepsilon^t$ from the proof of Theorem~\ref{theo:finite}, we can construct a Nash equilibrium that provides co-optimal deterministic strategies for both players, such that the co-optimal strategies contain at most one strategy switch on $I_\varepsilon^t$.
\end{lemma}

\begin{proof}
We describe the technique to find a constant co-optimal strategy on the right $\varepsilon$-environment $R_\varepsilon^t = (t,t+\varepsilon)\cap [0,t_{\max}]$ of $t$.

We write a constant strategy as $D = S+ R$ that is composed of the actions chosen by the safety player on $\locations_s$, and the actions chosen by the reachability player on $\locations_r$.
For this simple structure, we introduce a strategy improvement technique on the finite domain of deterministic choices for the respective player.

For a fixed strategy $S$ of the safety player, we can find an optimal counter strategy $R(S)$ of the reachability player by applying the technique described in Theorem~\ref{theo:finite}.
(For equivalent strategies, we make an arbitrary but fixed choice.)

We call the resulting vector $(-\dot{\prob}{}^{\M}_{S+R(S)}(l,t+\frac{1}{2}\varepsilon) \mid l \in \locations)$ the \emph{quality vector of $S$}.
Now, we choose an arbitrary $\overline{S}$ for which this vector is minimal.
(Note that there could, potentially, be multiple incomparable minimal elements.)

We now show that the following holds for $\overline{S}$ and all $\tau \in R_\varepsilon^t$:
$$-\dot{\prob}{}^{\M}_{\overline{S}+R(\overline{S})}(l,\tau) = \max_{a\in\act(l)}\sum_{l'\in\locations} \ratematrix(l,a,l') \cdot \big(\prob^{\M}_{\overline{S}+R(\overline{S})}(l',\tau) - \prob^{\M}_{\overline{S}+R(\overline{S})}(l,\tau)\big)$$
for all  $l\in \locations_r$, and
$$-\dot{\prob}{}^{\M}_{\overline{S}+R(\overline{S})}(l,\tau) = \min_{a\in\act(l)}\sum_{l'\in\locations} \ratematrix(l,a,l') \cdot \big(\prob^{\M}_{\overline{S}+R(\overline{S})}(l',\tau) - \prob^{\M}_{\overline{S}+R(\overline{S})}(l,\tau)\big)$$
for all  $l\in \locations_s$.
(Note that the order between the derivation is maintained on the complete right $\varepsilon$-environment $R_\varepsilon^t$.)

The first of these claims is a trivial consequence from the proof of Theorem~\ref{theo:finite}. (The result is, for example, the same if we had a single player CTMDP that, in the locations $\locations_s$ of the safety player, has only one possible action: the one chosen by $\overline{S}$.)

Let us assume that the second claim does not hold. Then we choose a particular $l\in \locations_s$ where it is violated.
Let us consider a slightly changed setting, in which the choices in $l$ are restricted to two actions, the action $a_1$ chosen by $\overline{S}$, and the minimising action $a_2$.
Among these two, one maximises, and one minimises
$$-\dot{\prob}{}^{\M}_{\overline{S}+R(\overline{S})}(l,\tau) = \min_{a\in\{a_1,a_2\}}\sum_{l'\in\locations} \ratematrix(l,a,l') \cdot \big(\prob^{\M}_{\overline{S}+R(\overline{S})}(l',\tau) - \prob^{\overline{S}+R(\overline{S})}(l,\tau)\big).$$

Let us fix all other choices of $S$, and allow the reachability player to choose among $a_1$ and $a_2$ (we `pass control' to the other player).
As shown in Theorem~\ref{theo:finite}, she will select an action that produces the well defined set of $\max$ equations for the resulting single player game.
Hence, choosing $a_1$ and keeping all other choices from $R(\overline{S})$ is the optimal choice for the reachability player in this setting
(as the $\max$ equations are satisfied, while they are dissatisfied for $a_2$).

Consequently, the quality vector for $\overline{S}$ is strictly greater than the one for the adjusted strategy.
That is, assuming that choosing an arbitrary maximal element does not lead to a satisfaction of the $\min$ and $\max$ equations leads to a contradiction.

We can argue symmetrically for the left $\varepsilon$-environment.
Note that the satisfaction of the $\min$ and $\max$ equations implies that it does not matter if we change the r\^ole of the safety and reachability player in our argumentation. \qed
\end{proof}

This lemma can easily be extended to construct simple co-optimal strategies:

\begin{theorem}
\label{theo:fingame}For CTMGs without discrete locations, there are cylindrical deterministic timed-positional co-optimal strategies for the reachability and the safety player.
\end{theorem}

\begin{proof}
First, Lemma~\ref{lem:epsilon} provides us with an open coverage of co-optimal strategies that switch at most once, and we can build a strategy that switches at most finitely many times from a finite sub-cover of the open space $[0,t_{\max}]$.
This strategy is everywhere locally co-optimal, and forms a Nash equilibrium:

It is straight forward to cut the interval $[0,t_{\max}]$ into a finite set of sub-intervals $[0,t_0]$, $(t_0,t_{1}]$, \ldots, $(t_{n-1},t_n]$ with $t_n=t_{\max}$, such that the strategy for the safety player is constant in all of these intervals.
We can use the construction from Theorem~\ref{theo:finite} (note that the proof of Theorem~\ref{theo:finite} does not use that the differential equations are initialised to $0$ or $1$ at $t_{\max}$) to construct an optimal strategy for the reachability player:
We can first solve the problem for the interval $[t_{n-1},t_n]$, then for the interval $[t_{n-2},t_{n-1}]$ using ${f}_{\mathsf{opt}}(l,t_{n-1})$ as initialisation, and so forth.
A similar argument can be made for the other player.

This provides us with the same differential equations, namely:
$$-\dot{f}_{\mathsf{opt}}(l,t) =\max\limits_{a\in\act(l)} \sum\limits_{l'\in\locations} \ratematrix(l,a,l') \cdot \left(f_{\mathsf{opt}}(l',t) -f_{\mathsf{opt}}(l,t)\right)$$
for $t\in [0,t_{\max}]$ and $l \in \locations_r$, and
$$-\dot{f}_{\mathsf{opt}}(l,t) =\min\limits_{a\in\act(l)} \sum\limits_{l'\in\locations} \ratematrix(l,a,l') \cdot \left(f_{\mathsf{opt}}(l',t) -f_{\mathsf{opt}}(l,t)\right)$$
for $t\in [0,t_{\max}]$ and $l \in \locations_s$.

Note that all Nash equilibria need to satisfy these equations (with the exception of $0$ sets, of course), because otherwise one of the players could improve her strategy.
\qed
\end{proof}

The extension of these results to the full class of CTMGs is straight forward: We would first reprove Theorem~\ref{theo:discrete} in the style of the proof of Theorem~\ref{theo:finite} (which requires to establish the Theorem in the first place).
The only extension is that we additionally get an equation $\prob^D_l(\tau) = \sum_{l'\in\locations} \probabilitymatrix(l,a_l,l') \cdot\prob_{l'}^D(\tau)$ for every discrete location $l$.
The details are moved to Appendix~\ref{app:reproof}.

\begin{theorem}
For continuous-time Markov Games, there are cylindrical deterministic timed-positional co-optimal strategies for the reachability and the safety player.
\end{theorem}

As a small side result, these differential equations show us that we can, for each continuous location $l_c \in \locations_c$ and every action $a \in \act(l_c)$, add arbitrary values to $\ratematrix(l_c,a,l_c)$ without changing the bounded reachability probability for every pair of schedulers. (Only if we change $\ratematrix(l_c,a,l_c)$ to $0$ we have to make sure that $a$ is not removed from $\act(l_c)$.)
In particular, this implies that we can locally and globally uniformise a continuous-time Markov game if this eases its computational analysis. (Cf.\ \cite{NeuhausserDelayedNondeterminism09} for the simpler case of CTMDPs.)

\section{Variances}\label{sec:variances}

In this section, we discuss the impact of small changes in the setting, namely the impact of infinitely many states or actions, and the impact of introducing a non-absorbing goal region. 

\paragraph{Infinitely Many States. }
If we allow for infinitely many states, optimal solutions may require infinitely many switching points.
To see this, it suffices to use one copy of the CTMDP from Figure \ref{fig:automatonAndReachability}, but with rates $i$ and $2i$ for the $i$-th copy,
and assign an initial probability distribution that assigns a weight of $2^{-i}$ to the initial state $A_i$ of the $i$-th copy.
(If one prefers to consider only systems with bounded rates, one can choose rates $1+\frac{1}{i}$ and $2+\frac{2}{i}$ .)
The switching points are then different for every copy, and an optimal strategy has to select the correct switching point for every copy.

\paragraph{Infinitely Many Actions. }

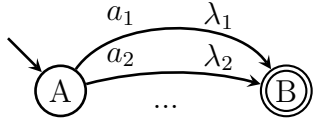
\begin{figure}[t]
\vspace{-.3cm}
\centering
\large
\begin{pspicture}[showgrid=false](-0.7,0)(2.5,3)
	\psset{arrowsize=5pt,nodesep=0pt,arrowlength=1,linewidth=1pt}
	\cnode[](0,1){10pt}{1}
	\rput(0,1){A}
	\psline[]{->}(-.7,1.7)(-.25,1.25)
	\cnode[linewidth=0.8pt](3,1){10pt}{goal}
	\cnode[linewidth=0.8pt](3,1){8pt}{goalinner}
	\rput(3,1){B}
	\nccurve[angleA=55,angleB=125]{->}{1}{goal}
	\nccurve[angleA=15,angleB=165]{->}{1}{goal}
	\rput(2.1,2.0){$\lambda_1$}
	\rput(0.8,2.0){$a_1$}
	\rput(2.1,1.45){$\lambda_2$}
	\rput(0.8,1.45){$a_2$}
	\rput(1.4,.8){...}
\end{pspicture}
\normalsize
\caption{An example CTMDP with infinitely many actions.}%
\label{fig:inftyactions}%
\end{figure}

If we allow for infinitely many actions, there is not even an optimal strategy if we restrict our focus to CTMDPs with two locations, an initial location and an absorbing goal location.
For the CTMDP of Figure~\ref{fig:inftyactions} with the natural numbers $\mathbb N$ as actions and rate $\lambda_i=2-\frac{1}{i}$ for the action $i \in \mathbb{N}$ if we have a reachability player and $\lambda_i=\frac{1}{i}$ if we have a safety player, every strategy $\S$ can be improved over by a strategy $\S'$ that always chooses the successor $i+1$ when of the action $i$ chosen by $\S$.

\paragraph{Reachability at $t_{\max}$.}
If we drop the assumption that the goal region is absorbing, one might be interested in the marginally more general problem to be (not to be) in the goal region at time $t_{max}$ for the reachability player (safety player, respecively).
For this generalisation, no substantial changes need to be made:
It suffices to replace
$$ f_{\mathsf{opt}}(l,t) = \prob_{\S}^{\M}(l,t) = 1 \qquad \mbox{ for all goal locations $l \in G$ and all }t \leq t_{\max}$$
by
$$ f_{\mathsf{opt}}(l,t_{\max}) = \prob_{\S}^{\M}(l,t_{\max}) = 1 \qquad \mbox{ for all goal locations }l \in G.$$

(In order to be flexible with respect to this condition, the $-\dot{f}_{\mathsf{opt}}(l,t)$ are defined for goal locations as well.
Note that, when all goal locations are absorbing, the value of $-\dot{f}_{\mathsf{opt}}(l,t)$ is $0$ and ${f}_{\mathsf{opt}}(l,t)$ is $1$ for all goal locations $l\in G$ and all $t\in [0,t_{\max}]$.)


\newpage
\appendix
\section*{\Large Appendix}

As to be expected by the topic, the paper is based in large parts on measure theory.
While the techniques are standard and straight forward for the experts in the field we provide a short introduction to the ideas exploited;
while we do not use any technique beyond the standard curriculum of a math degree, we assume that some recap of the ideas behind the completion of metric spaces and its application in our measures in Section~\ref{sec:complete}.
However, Section~\ref{sec:complete} is but a short introduction to the ideas, and cannot serve as a self contained introductory to the techniques.

Section~\ref{app:optimal} contains a short recap on the differential equations that describe $f_{\min}$ and $f_{\max}$, and, more generally, the development of $\prob_{\S}^{\M}(l,t)$ in Subsection~\ref{app:diff}, and a proof that $f_{\max}$ truly establishes an upper bound on the performance of any measurable scheduler in Subsection~\ref{app:proof}, which constitutes a proof of Lemma~\ref{lem:stupid}.

\section{Completion of Metric Spaces}
\label{sec:complete}
A metric space is called complete if every Cauchy sequence in it converges.
A Cauchy sequence in a metric space $(M,d)$ is a sequence $s:\mathbb N \rightarrow M$ such that the following holds:
$$\forall \varepsilon > 0\ \exists n \in \mathbb N\ \forall l,m > n.\ d\big(s(l),s(m)\big)<\varepsilon$$

Intuitively one could say that a Cauchy sequence converges, but not necessarily to a point within the space. For example, a sequence of rational numbers that converges to $\sqrt{2}$ is a converging sequence in the real numbers, but not in the rationals---as the limit point is outside of the carrier set---but it is still a Cauchy sequence.
\\

The basic technique to complete an incomplete metric space $(M,d)$ is to use the Cauchy sequences of this space as the new carrier set $M'$, and define a distance function $d'$ between two Cauchy sequences $s,s'\in M'$ of $M$ to be $d'(s,s')=\lim\limits_{n \rightarrow \infty}d\big((s(n),s'(n))\big)$.
Now, $(M',d')$ is not yet a metric space, because two different Cauchy sequences---for example the constant $0$ sequence and the sequence $s(n)=\frac{1}{2^n}$ of the rationals---can have distance $0$.

Technically, one therefore defines equivalence classes of Cauchy sequences that have distance $0$ with respect to $d'$ as the new carrier set $M''$ of a metric space $(M'',d'')$, where the distance function $d''$ is defined by using $d'$ on representatives of the quotient classes of Cauchy sequences.
This also complies with the intuition: a Cauchy sequence in $M$ is meant to represent its limit point (which is not necessarily in $M$), and hence two Cauchy sequences with the same limit point should be identified.

The resulting metric space $(M'',d'')$ is complete by construction. The simplest example of such a completion is the completion of the rational numbers into the real numbers.
And on this level, a straight forward effect of completion can be easily explained:
To end up with $(M'',d'')$ (or a space isomorphic to it), we can start with any \textbf{dense} subset $S$ of $M''$.

A subset $S\subset M''$ is dense in $M''$ if, for every point $m\in M''$ and every $\varepsilon >0$, there is a point $s\in S$ with $d''(m,s)<\varepsilon$.
Looking at the definition, one immediately sees the connection to Cauchy sequences: One could intuitively say that $S\subset M''$ is dense in $M''$ if, for every point $m\in M''$, there is a Cauchy sequence $s$ with limit point $m$.

Hence, it does not matter which dense set we use as a starting point. Of course, it works to use $M''$, in the example of the real numbers, we can start with the real numbers themselves without gaining anything by applying the completion twice, we can start with the transcendent numbers, the non-transcendent numbers, or, more down to earth, with finite decimal fractions. Note that a subset is dense in $M''$ if it is dense in some $S$ that is dense in $M''$.

\subsection{Application in Measure Theory}
Another famous application of this completion technique is the completion of Riemann integrable functions to Lebesgue integrable functions. The difference metrics between two functions is the Riemann integral over the absolute value of their difference.
(Strictly speaking, this does again not form a metric space, and we again have to use the quotient class of functions with difference $0$.)

Riemann integrable functions do, for example, allow only for bounded functions, but there are other problems as well; for example, we cannot integrate over the characteristic function of the rational numbers. (The wikipedia article to Riemann integrable functions is nice and gives a good overview on the weaknesses.)

Using the completion technique defined above, one can, for example, integrate over the characteristic function of the rationals by enumerating them, that is, by defining a surjection $s: \mathbb N \rightarrow \mathbb Q$, and choose $f_i$ to be the Riemann integrable function that is $1$ at the mapping $s(\{1,2,\ldots,i\})$ of the initial sequence of length $i$ of the naturals. Clearly, the limit of the sequence $f_1,f_2,f_3,\ldots$ is the characteristic function, the Riemann integral over all $f_i$ is $0$ (no matter over which interval we integrate) the sequence is a Cauchy sequence.


The completion of the space of Riemann integrable functions (which essentially establishes the Lebesgue integrable functions) is space of all Cauchy sequences of Riemann integrable functions (or again representatives of the quotient classes of equivalent Cauchy sequences). 

Again, it does not make a difference if we start the completion with the Riemann integrable functions, or with weaker concepts, as long as they are dense in the resulting space, or indeed in the Riemann integrable functions.
A well known example for such a class is the class of block functions, where the value changes only in finitely many positions.

\subsection{Application in Our Measure}
\label{app:ourmeasure}
The definition of the measures for continuous-time Markov chains, games (with fixed strategies), and decision processes (with a fixed scheduler) works in exactly this way:
For Markov chains, one defines the probability measure for simple disjoint (or almost disjoint) sets, for which the probability is simple to determine, and such that one can define the probability of reaching the goal region in time for cylindrical sets.

When considering games and decision processes without fixed scheduling policies, however, we have two layers of completions---one layer \emph{for a given} cylindrical schedulers, and one \emph{on} cylindrical schedulers. 

\paragraph{\bf Measure for a given cylindrical schedulers.}
In the case of games and decision processes, one starts to define it for a particularly friendly and easy to handle class of strategies or schedulers, respectively, such that the techniques from Markov chains can be extended with minor adjustments.
Our cylindrical schedulers---which in our paper also represent pairs of strategies in games---with only finitely many switching points are an example of such an extension.

The measure \emph{for} a given cylindrical scheduler is a mild extension of the techniques for Markov chains.
The building blocks of the probability measure define the probability on \emph{cylindrical sets}, and they are a straight forward extension of the probability measure for continuous time Markov chains to continuous time Markov decision processes (and games) with a fixed scheduler of this type.
However, they only describe the likelihood for cylindrical sets, and without completing the space we can but use finite sums over disjoint sets.

To obtain the time bounded reachability probability, we use Cauchy sequences of such sets that converge against all sets of paths on which a goal region is reached.
A representative of this equivalence class would be a sequence $P_1, P_2,P_3,\ldots$ or sets of paths, where 
$P_i$ contains the cylindrical sets of length up to $i$ in which the goal region is reached.
This is a Cauchy sequence (cf.~the argument in Subsection~\ref{app:proof}), and the limit contains all finite paths upon which the goal region is reached.

\paragraph{\bf Measure for a given cylindrical schedulers.}
While we have established a measure for a given cylindrical scheduler, the class of cylindrical schedulers is not particularly strong, and, like with Riemann integrable functions, we need to strengthen the class of schedulers we allow for.
In order to exploit the aforementioned completion technique, we have to create a suitable metric space on them, and in order to introduce such a metric space, we need a measure for the difference between strategies.
Such a measure reflects the likelihood that two different strategies ever lead to different actions.
For example, for \emph{deterministic} schedulers $D$ and $E$ for a CTMDP $\M$ we define a difference scheduler $\delta_{\{D,E\}}$ that uses the actions of $D$/$E$ on every history, on which they coincide, and a fresh action $(a_D,a_E)$ if $D$ chose $a_D$ and $E$ chose $a_E \neq a_D$ upon this history.
(Note that the difference measure uses a slightly adjusted CTMDP $\M'$; it also has a fresh goal location.)

The new action $(a_D,a_E)$ leads to a fresh continuous goal location $g$; the old goal locations do not remain goal locations, the new goal region contains only $g$.

For a continuous location $l$, we fix
\begin{itemize}
 \item $\ratematrix(l,(a_D,a_E),g)=\ratematrix(l,a_D,\locations)+\ratematrix(l,a_E,\locations)$ and
 \item $\ratematrix(l,(a_D,a_E),l')=0$ for all locations $l'\neq g$
\end{itemize}
for the new actions, and maintain the entries to in the rate matrix for the old actions.
(To be formally correct, we provide $g$ with only one enabled action $a$---with $\ratematrix(g,a,l')=0$ for all locations $l'\neq g$---that $\delta_{\{D,E\}}$ selects in $g$ upon any history.)

For discrete locations, we would fix $\probabilitymatrix(l,(a_D,a_E),g)=1$ (and $\probabilitymatrix(l,(a_D,a_E),l')=0$ for $l'\neq g$) for all new actions, and maintain the entries for the old actions in $\probabilitymatrix$.

The distance between two schedulers is then defined as the likelihood that $g$ is reached in the adjusted CTMDP within the given time bound when using the \emph{cylindrical} scheduler $\delta_{\{D,E\}}$.

\paragraph{\bf The metric space on cylindrical schedulers.}
This distance function almost defines a metric space on the cylindrical schedulers we used as a basic building block; symmetry, triangle inequation and non-negativity obviously hold.
However, we again have to resort to a carrier set of quotient classes of schedulers with distance $0$, in order to satisfy $d(x,y) \Leftrightarrow x=y$. (One can think of different actions on unreachable paths.)

\paragraph{\bf Effect of the metric.}
The distance is defined in a way that guarantees that the absolute value of the difference between the time bounded reachability probabilities for $D$ and $E$ is bounded by the distance between these schedulers (or representatives of their quotient class):
To see this, it suffices to look at the scheduler $\delta_{\{D,E\}}$ in $\M$ with two adjusted goal regions:
\begin{enumerate}
 \item If we use the old goal region plus $g$ as our new goal region, then obviously the time bounded reachability is better than the time bounded reachability of $D$ and $E$ in $\M$.
\item  If we use the old goal region (but not plus $g$---$g$ rather becomes a non-accepting sink---as a new goal region, then the time bounded reachability is worse than the time bounded reachability of $D$ and $E$ in $\M$.
\end{enumerate}

The difference between the two cases, however, is exactly the \emph{distance} between $D$ and $E$, which therefore in particular is an upper bound on the difference between their time bounded reachability probability.

\paragraph{\bf Completion.}
The time bounded reachability of a Cauchy sequence can therefore be defined as the limit of the time bounded reachability of the elements of the sequence, which is guaranteed to exist (and to be unique) by the previous argument.

\paragraph{\bf Extension to randomised schedulers.}
If the schedulers $D$ and $E$ from above are randomised, then we would define the distribution chosen by $\delta_{\{D,E\}}$ in line with the definition from above. We first fix an arbitrary global order on all actions used as tie breaker in our construction.

\begin{itemize}
\item If, in a particular history that ends in some location $l$, $D$ chooses an action $a$ with probability $p_D^a$ and $E$ chooses $a$ with probability $p_E^a$ then $\delta_{\{D,E\}}$ chooses $a$ with probability $\min\{p_D,p_E\}$

\item If the probabilities assigned by this rule sum up to one for $\delta_{\{D,E\}}$ on this history is thus defined.
\item Otherwise, if $l$ is continuous,  we choose an action $a_D$ among the actions with $p_D^a > p_E^a$ that maximises $\ratematrix(l,a,\locations)$, using the fixed global order as a tie breaker.
For discrete locations, we use the tie breaker only.

\item We then accordingly choose an action $a_E$ among the actions with $p_E^a > p_D^a$ that maximises $\ratematrix(l,a,\locations)$, using the fixed global order as a tie breaker.
For discrete locations, we again use the tie breaker only.
\item We assign the remaining probability weight to the decision $(a_D,a_E)$.
\end{itemize}

All arguments from above extend to this case.

\paragraph{\bf Measurable schedulers.}
The term \emph{measurable scheduler} refers to the limit scheduler of such a Cauchy sequence.
To restrict the attention to this class of schedulers is simply a requirement caused by the definition of the measure:
Only for such schedulers the time bounded reachability probability is defined.

However, we do not think this is a real drawback, as the set of measurable schedulers far outreaches the power of everything one might have in mind when talking about schedulers (like choosing $a$ on the points in time defined by the Cantor set), and its not easy to describe a non-measurable scheduler in the first place.

\section{Optimal reachability probability} 
\label{app:optimal}
\subsection{Differential Equations}
\label{app:diff}
The differential equations defining $f_{\mathsf{opt}}$ are simply the differential equations in place when a strategy is locally constant. This holds almost everywhere (everywhere but in a $0$-set of positions) in case of the \emph{cylindrical} schedulers that are the basic building blocks in the incomplete space that we have completed by considering Cauchy sequences of cylindrical schedulers.

Hence, for every cylindrical scheduler we can partition the interval $[0,t_{\max}]$ into a finite set of intervals $I_0$, $I_1$, $I_2$, $\ldots$, $I_n$ as described in the preliminaries.

Within such an interval,
\[-\dot{\prob_{\S}}(\pi,t) = \sum_{l'\in\locations} \ratematrix(l,a,l') \cdot \left(\prob_{\S}(\pi\xrightarrow{a,t} l',t) -\prob_{\S}(\pi,t)\right)\quad \mbox{ for }t\in I_i\]
holds for discrete schedulers, where $\pi$ is a timed path that ends in $l$, $a$ is the deterministic choice the scheduler makes in $I_i$ on this history, and $\pi\xrightarrow{a,t} l'$ is its extension.
%
For randomised schedulers, 
\[-\dot{\prob_{\S}}(\pi,t) = \sum_{a \in \act}h(a)\sum_{l'\in\locations} \ratematrix(l,a,l') \cdot \left(\prob_{\S}(\pi\xrightarrow{a,t} l',t) -\prob_{\S}(\pi,t)\right)\quad \mbox{ for }t\in I_i\]
holds, where $\pi$ is a timed path that ends in $l$, $h(a)$ is the likelihood that the cylindrical scheduler makes the decision $a$ in $I_i$ on this history, and $\pi\xrightarrow{a,t} l'$ is its extension.

To initialise the potentially infinite set of differential equations, we have the following initialisations:
\begin{itemize}
 \item [$\circ$] $\prob_{\S}(\pi,t) = 1$ holds for all timed histories $\pi$ that contain (and hence end up in) locations $l \in G$ in the goal region and all $t \leq t_{\max}$,
 \item [$\circ$] $\prob_{\S}(\pi,t_{\max}) = 0$ holds for all timed histories $\pi$ that contain only non-goal locations $l \notin G$, and
 \item [$\circ$] $\prob_{\S}(\pi,t) = 0$ holds for all locations $l \in \locations$ and all $t > t_{\max}$.
\end{itemize}

Additionally, we have to consider what happens at the intersection $t_i$ of the fringes of $I_i$ and $I_{i+1}$ for $0\leq i <n$.
But obviously, we can simply first solve the differential equations for $I_n$, then use the values of $f(\pi,t_{n-1})$ as initialisations for the interval $I_{n-1}$, and so forth.
\\

\noindent
\textbf{Remark:} For timed positional deterministic schedulers we get
\[-\dot{\prob_{\S}}(l,t) = \sum_{l'\in\locations} \ratematrix(l,a,l') \cdot \left(\prob_{\S}(l',t) -\prob_{\S}(l,t)\right)\quad \mbox{for }t\in I_i\mbox{, and}\]
\[-\dot{\prob_{\S}}(l,t) = \sum_{a \in \act}h(a)\sum_{l'\in\locations} \ratematrix(l,a,l') \cdot \left(\prob_{\S}(l',t) -f(l,t)\right)\quad \mbox{for }t\in I_i\]
for timed positional randomised schedulers. In both cases, the initialisation reads
\begin{itemize}
 \item [$\circ$] $\prob_{\S}(l,t) = 1$ holds for all goal locations $l \in G$ and all $t \leq t_{\max}$,
 \item [$\circ$] $\prob_{\S}(l,t_{\max}) = 0$ holds for all non-goal locations $l \notin G$, and
 \item [$\circ$] $\prob_{\S}(l,t) = 0$ holds for all locations $l \in \locations$ and all $t > t_{\max}$.
\end{itemize}
\smallskip

Obviously, these differential equations can also be used in the limit.

\subsection{Timed positional schedulers suffice for optimal time-bounded reachability}
\label{app:proof}

In this subsection we sketch a proof of a variant of Lemma~\ref{lem:stupid};
we demonstrate the following claim for arbitrary $t_{\max}\geq 0$:

\begin{lemma}
For a CTMDP $\M$ with only continuous locations, $\prob_{\S}^{\M}(l,t) \leq f_{\max}(l,t)$ holds for every scheduler $\S$, every location $l$, and every $t\in [0,t_{\max}]$.
\end{lemma}

In the proof, we assume a scheduler that provides an $3\varepsilon$ better result,
and then sacrifice one $\varepsilon$ to transfer to cylindrical schedulers (going back to the simpler incomplete space of cylindrical schedulers, but with completed reachability measure),
and then sacrificing a second $\varepsilon$ to discard long histories from consideration (going back to the simple space of finite sums over cylindrical sets).

As a result, we can do the comparison in a simple finite structure.

\begin{proof}
Let us assume that the claim is incorrect.
Then, there is a CTMDP $\M$ with location $l_0$ and a scheduler $\S_{3\varepsilon}$ for $\M$ such that the time bounded reachability probability is at least $3\varepsilon$ higher for some $\varepsilon >0$ and $t_0\in[0,t_{\max}]$.
That is, $\prob_{\S_{3\varepsilon}}^{\M}(l_0,t_0) - f_{\max}(l_0,t_0) > 3 \varepsilon$

Let us fix appropriate $\M$, $l_0$, and $\S_{3\varepsilon}$.
(Note that $\S_{3\varepsilon}$ does not have to be timed positional or deterministic.)

Recall that $\S_{3\varepsilon}$ is the limit point of a Cauchy sequence of cylindrical schedulers.
Hence, almost all of these cylindrical schedulers have distance $<\varepsilon$ to $\S_{3\varepsilon}$.
\\

Let us fix such a cylindrical scheduler $\S_{2\varepsilon}$ with distance $<\varepsilon$ to $\S_{3\varepsilon}$.
The time bounded reachability probability of $\S_{2\varepsilon}$ is still at least $2\varepsilon$ higher compared to $f_{\max}$.
That is,  $\prob_{\S_{2\varepsilon}}^{\M}(l_0,t_0) - f_{\max}(l_0,t_0) > 2 \varepsilon$ holds true.

For $\S_{2\varepsilon}$, we now consider a tightened form of time bounded reachability, where we additionally require that the goal region is to be reached within $n_\varepsilon$ steps.
We choose $n_\varepsilon$ big enough that the likelihood of seeing more than $n_\varepsilon$ discrete events is less than $\varepsilon$. We call this time bounded $n_\varepsilon$ reachability.

\noindent
\textbf{Remark:} We can estimate $n_\varepsilon$ by taking the maximal transition rate $\lambda_{\max} = \max\{\ratematrix(l,a,\locations) \mid l \in \locations_c, a \in \act\}$, and choose $n_\varepsilon$ big enough that the likelihood of having more than $n_\varepsilon$ transitions was smaller than $\varepsilon$ even if all transitions had transition rate $\lambda_{\max}$.
As the number of steps is Poisson distributed in this case,  a suitable $n_\varepsilon$ is easy to find.
\\

The adjustment to time bounded $n_\varepsilon$ reachability leads to a small change in the initialisation of the differential equations:
For timed histories $\pi$ of length $>n_\varepsilon$ that do not contain a location $l \in G$ in the goal region within the first $n_\varepsilon$ steps, we use $f(\pi,t) = 0$ (even if it contains a goal region after more than $n_\varepsilon$ steps) for all $t \in [0,t_{\max}]$.
As the probability measure of all timed histories $\pi$ of length $>n_\varepsilon$ is $<\varepsilon$, time bounded $n_\varepsilon$ reachability for $\S_{2\varepsilon}$ is still at least $\varepsilon$ higher than the value for $f_{\max}$.

Let us use $f(\pi,t)$ to express the time bounded $n_\varepsilon$ reachability for $\S_{2\varepsilon}$ on a path $\pi$ at time $t$.
Then this claim can be phrased as
$f(l_0,t_0) - f_{\max}(l_0,t_0) > \varepsilon.$
\\

We have now reached a finite structure, and can easily show that this leads to a contradiction:
We provide an inductive argument which even demonstrates that $f_{\max}(l,t) \geq f(\pi,t)$ holds for all $\pi$ that end in $l$ and all $t\in [0,t_{\max}]$.
\\

As a basis for our induction, this obviously holds for all timed histories longer than $n_\varepsilon$: in this case, $f_{\max}(l,t)=1$ or $f(\pi,t)=0$ holds true (where the or is not exclusive).
\\

For our induction step, let us assume we have demonstrated the claim for all histories of length $>n$.
Let us, for a timed history $\pi$ of length $n$ that ends in $l$ and some point $t \in [0,t_{\max}]$ assume that $f_{\max}(l,t) \leq f(\pi,t)$.

For $l \in G$ the initialisation conditions immediately lead to the contradiction $1<f(\pi,t)$. 
For $l \notin G$, we can stepwise infer
\begin{description}
 \item[$-\dot{f}_{\max}(l,t) =$] $\max \big\{\sum_{l'\in\locations} \ratematrix(l,a,l') \cdot \left(f_{\max}(l',t) -f_{\max}(l,t)\right) \mid a \mbox{ is enabled in }l\big\}$
 \item[$\quad\geq $] $\sum_{distribution} \sum_{l'\in\locations} \ratematrix(l,a,l') \cdot \left(f_{\max}(l',t) -f_{\max}(l,t)\right)$
 \item[$\quad\geq$] $\sum_{distribution} \sum_{l'\in\locations} \ratematrix(l,a,l') \cdot \left(f_{\max}(l',t) -f(\pi,t)\right)$ \hfill (with $f_{\max}(l,t) \leq f(\pi,t)$)
 \item[$\quad\geq$] $\sum_{distribution} \sum_{l'\in\locations} \ratematrix(l,a,l') \cdot \left(f(\pi \xrightarrow{a,t} l',t) -f(\pi,t)\right)$ \hfill (with I.H.)
 \item[$\quad=$] $-\dot{f}(\pi,t)$.
 \end{description}

Taking into account that $f_{\max}(l,t_{\max})$ and $f(\pi,t_{\max})$ are both initialised to $0$ for $l \notin G$, we can, using the just demonstrated
$f_{\max}(l,t) \leq f(\pi,t) \Rightarrow \dot{f}_{\max}(l,t) \leq \dot{f}_(\pi,t)$,
infer $f_{\max}(l,t) \geq f(\pi,t)$ for all $t\in [0,t_{\max}]$: This inequation holds on the right fringe of the interval (initialisation), and when we follow the curves of $f(l,t)$ and $f_{\max}(l,t)$ to the left along $[0,t_{\max}]$, then every time $f$ would catch up with $f_{\max}$, $f$ cannot fall steeper than $f_{\max}$ (where `fall' takes the usual left-to-right view, in the right-to-left direction we consider one should maybe say `cannot have a steeper ascend') at such a position, and hence cannot not get above $f_{\max}$.
\\

In particular, $f(l_0,t_0)\leq f_{\max}(l_0,t_0)$, which contradicts the initial assumption.
\end{proof}

The $\min$ case can be proven accordingly, which provides a full proof of Lemma~\ref{lem:stupid}.

(Note that the extension to CTMDPs with both discrete and continuous locations is provided in Section~\ref{sec:discrete}.)

\section{Reproof of Theorem \ref{theo:discrete}}
\label{app:reproof}

To lift Theorem~\ref{theo:fingame} to the full class of CTMGs, we reprove Theorem~\ref{theo:discrete} in the style of the proof of Theorem~\ref{theo:finite}.
Recall that Theorem~\ref{theo:finite} establishes the existence of an optimal cylindrical scheduler \emph{using} the existence of an optimal measurable scheduler, and the form of the (differential) equations defining the time-bounded reachability probability for it.
The proof given in this appendix can therefore not been used to supersede the proof in the paper.

First we observe from the proof of Theorem~\ref{theo:discrete} that, for discrete locations $l\in \locations_d$, the equations
\[f_{\mathsf{opt}}(l,t) = \hspace*{-3mm}
{\begin{array}{c}
\vspace*{-4pt} \\
\mathsf{opt} \vspace*{-4pt}\\ \scriptsize \mbox{$a\in\act(l)$}
\end{array}}\sum_{l'\in\locations}  \probabilitymatrix(l,a_l,l') \cdot f_{\mathsf{opt}}(l',t) \mbox{ for }t\in [0,t_{\max}],\]
holds for $\mathsf{opt}\in \{\min,\max\}$,
and that they together with the differential equations for the continuous locations (the differential equations remain unchanged), define $f_{\mathsf{opt}}$.

The difference in the proof of Theorem~\ref{theo:dfinite} compared to the proof of Theorem~\ref{theo:finite} are marked in {\color{blue}blue}.
\begin{theorem}
\label{theo:dfinite}
For a single player continuous-time Markov game with only a reachability player, there is an optimal deterministic scheduler with finitely many switching points.
\end{theorem}

\begin{proof}
We have seen that the true optimal reachability probability is defined by a system of {\color{blue}equations and} differential equations.
In this proof we consider the effect of starting with the `correct' values for a time $t\in [0,t_{\max}]$, but
\emph{locally fix a positional strategy} for a small left or right $\varepsilon$-environment of $t$.
That is, we consider only schedulers that keep their decision constant for a (sufficiently) small time $\varepsilon$ before or after $t$.

Given a CTM{\color{blue}G} $\M$, we consider the {\color{blue}equations and} differential equations that describe the development of the reachability probability for each location $l$ under a positional deterministic strategy $D$:
\[
-\dot{\prob}{}^D_l(\tau) = \sum_{l'\in\locations} \ratematrix(l,a_l,l') \cdot \big(\prob_{l'}^D(\tau)-\prob_{l}^D(\tau)\big)\quad \color{blue} \mbox{ for }l\in \locations_c,\vspace{-1mm}
\]
\[\color{blue}
{\prob}{}^D_l(\tau) = \sum_{l'\in\locations} \probabilitymatrix(l,a_l,l') \cdot \prob_{l'}^D(\tau) \quad \mbox{ for }l\in \locations_d,\vspace{-1mm}
\]
where $a_l$ is the action chosen at $l$ by $D$, 
 starting at the support point $f_{\max}(l,t)$.

Different to the development of the true probability, the development of these linear differential equations provides us with smooth functions.
This provides us with more powerful techniques when comparing two locally positional strategies:
Each deterministic scheduler defines a system $\dot{y} = A y$ of ordinary homogeneous linear differential equations with constant coefficients.

As a result, the solutions $\prob^D_l(\tau)$ of these differential equations---and hence their differences $\prob^{D'}_l(\tau)-\prob^D_l(\tau)$---can be written as finite sums $\sum_{i=1}^n P_i(\tau) e^{\lambda_i \tau}$, where $P_i$ is a polynomial and the $\lambda_i$ may be complex.
Consequently, these functions are holomorphic.

Using the identity theorem for holomorphic functions, $t$ can only be a limit point of the set of $0$ points of $\prob^{D'}_l(\tau)-\prob^D_l(\tau)$ if $\prob^{D'}_l(\tau)$ and $\prob^D_l(\tau)$ are identical on an $\varepsilon$-environment of $t$.
The same applies to their derivations: $\dot{\prob}{}^{D'}_l(\tau)-\dot{\prob}{}^D_l(\tau)$ either has no limit point in $t$, or $\dot{\prob}{}^{D'}_l(\tau)$ and $\dot{\prob}{}^D_l(\tau)$ are identical on an $\varepsilon$-environment of $t$.

For the remainder of the proof, we fix, for a given time $t$, a sufficiently small $\varepsilon>0$ such that, for each pair of schedulers $D$ and $D'$ {\color{blue} the following holds:
for} every location $l\in\locations_{\color{blue}c}$,  $\dot{\prob}{}^{D'}_l(\tau)-\dot{\prob}{}^{D}_l(\tau)$ is either $<0$, $=0$, or $>0$ on the complete interval $L_\varepsilon^t=(t-\varepsilon,t)\cap [0,t_{\max}] \ni \tau$, and, possibly with different sign, for the complete interval $R_\varepsilon^t=(t,t+\varepsilon)\cap [0,t_{\max}] \ni \tau${\color{blue}; 
and for every location $l\in\locations_d$,  ${\prob}{}^D_l(\tau)-{\prob}{}^{D'}_l(\tau)$ is either $<0$, $=0$, or $>0$ on the complete interval $L_\varepsilon^t=(t-\varepsilon,t)\cap [0,t_{\max}] \ni \tau$, and, possibly with different sign, for the complete interval $R_\varepsilon^t=(t,t+\varepsilon)\cap [0,t_{\max}] \ni \tau$.}

We argue the case for the left $\varepsilon$-environment $L_\varepsilon^t$.
In the `$>$' case for a location $l$, we say that $D$ is \emph{$l$-better} than $D'$.
We call $D$ \emph{preferable} over $D'$ if $D'$ is not $l$-better than $D$ for any location $l$, and \emph{better} than $D'$ if $D$ is preferable over $D'$ and $l$-better for some $l\in \locations$.

If $D'$ is $l$-better than $D$ in exactly a non-empty set $\locations_b\subset\locations$ of locations, then we can obviously use $D'$ to construct a strategy $D''$ that is better than $D$ by switching to the strategies of $D'$ in exactly the locations $\locations_b$.

Since we choose our strategies from a finite domain---the deterministic positional schedulers---this can happen only finitely many times.
Hence we can stepwise \emph{strictly} improve a strategy, until we have constructed a strategy $D_{\max}$ that is preferable over all others.

By the definition of being preferable over all other strategies, $D_{\max}$ satisfies
$$-\dot{\prob}{}^{D_{\max}}_l(\tau) = \max_{a\in\act(l)}\sum_{l'\in\locations} \ratematrix(l,a,l') \cdot \big(\prob_{l'}^{D_{\max}}(\tau) - \prob_{l}^{D_{\max}}(\tau)\big) \quad \mbox{ for all } \tau \in L_\varepsilon^t, l\in \locations_{\color{blue}c},$$
\[\color{blue}{\prob}{}^{D_{\max}}_l(\tau) = \max_{a\in\act(l)}\sum_{l'\in\locations} \probabilitymatrix(l,a,l') \cdot \prob_{l'}^{D_{\max}}(\tau) \quad \mbox{ for all } \tau \in L_\varepsilon^t, l\in \locations_{d}.\]

We can use the same method for the right $\varepsilon$-environment $R_\varepsilon^t$, and pick the decision for $t$ arbitrarily;
we use the decision from the respective left $\varepsilon$ environment.

Now we have fixed, for an $\varepsilon$-environment of an arbitrary $t\in [0,t_{\max}]$, an optimal scheduler with at most one switching point.
As this is possible for all points in $[0,t_{\max}]$, the sets $I_\varepsilon^t = L_\varepsilon^t \cup R_\varepsilon^t$ define an open cover of $[0,t_{\max}]$.
Using the compactness of $[0,t_{\max}]$, we infer a finite sub-cover, which establishes the existence of a strategy with a finite number of switching points.
\qed
\end{proof}

Again, the proof for single player safety games runs accordingly.

\begin{theorem}
\label{theo:dfinitea}
For a single player continuous-time Markov game with only a safety player, there is an optimal deterministic scheduler with finitely many switching points.
\end{theorem}

\section{From Late to Early Scheduling}
\label{app:late2early}
Our main motivation for introducing discrete transitions is not the slightly improved generality of the model (nice though it is as a side result), but the introduction of a framework that covers \emph{both} early schedulers (which have to fix an action when \emph{entering} a location), and the late schedulers used in the paper.

Late schedulers are naturally subsumed in our model, as the schedulers we assume are the more powerful late schedulers.
To embed early schedulers as well, it suffices to use a simple translation:
we `split' every continuous location $l_c$ into a fresh discrete location $l_c^d$, and one fresh continuous location $l_c^a$ for each action $a\in \act(l_c)$ enabled~in~$l_c$.

Every incoming transition to $l_c$ is re-routed to $l_c^d$, $l_c^d$ has an outgoing transition $a$ that surely leads to $l_c^a$ ($\probabilitymatrix(l_c^d,a,l_c^a)=1$) for each action $a\in \act(l_c)$ enabled in $l_c$, and no other outgoing transition.
In $l_c^a$, we have $\act(l_c^a)=\{a\}$, and the entries in  $\ratematrix(l_c^a,a,l)$ are the entries taken from $\ratematrix(l_c,a,l)$ for discrete locations $l$, and re-routed to the respective $l^d$ for continuous locations.
Probability mass assigned to $l_c$ is moved to $l_c^d$ by the translation, and if $l_c$ is a goal state, so are $l_c^d$ and the $l_c^a$'s.

Intuitively, every occurrence of $\xrightarrow{*,t} l_c \xrightarrow{a,t'}$ is replaced by  $\xrightarrow{*,t} l_c^d \xrightarrow{a,t} l_c^a \xrightarrow{a,t'}$;
$l_c$ is the beginning of the path, $l_c \xrightarrow{a,t'}$ is replaced by $l_c^d \xrightarrow{a,t} l_c^a \xrightarrow{a,t'}$.

Obviously, there is a trivial bijection between early schedulers for a thus translated CTMDP (or, indeed, CTMG), and the late schedulers in the mapping: the actions chosen in a discrete location are doubled, and there is no alternative to doubling it.

As a consequence, the existence of finite deterministic optimal control extends to early scheduling. 

\paragraph{\bf From Early to Late Scheduling. } As a side remark, we would like to point out that a similar translation can be used to reduce finding optimal control for late schedulers to finding optimal control for early schedulers.
Following up on the remark that---assuming late scheduling---we can work with the uniformisation of a CTMG when seeking co-optimal control, it suffices to establish such a translation for uniform CTMGs, that is, for CTMGs where the transition rate $\ratematrix(l_c,a,\locations)$ is constant for all continuous locations $l_c \in \locations$, and all their enabled actions $a \in \act(l_c)$.
And for such uniform CTMGs we can move the decision into a fresh discrete location \emph{after} the continuous location, just as we moved it to a fresh discrete location \emph{before} the continuous location in our reduction from late to early scheduling.

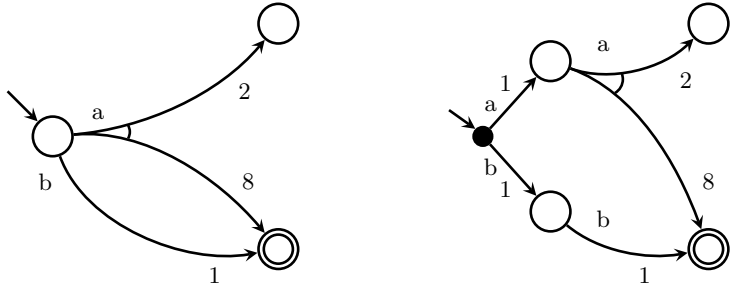
\begin{figure}[t]%
\centering
\hspace{.5cm}\begin{pspicture}[showgrid=false](0,0.3)(4,4)
	\psset{arrowsize=5pt,nodesep=0pt,arrowlength=1,linewidth=1pt}
	\cnode[](0,2){8pt}{1}
	\cnode[](3,3.5){8pt}{2}
	\cnode[linewidth=1pt](3,0.5){8pt}{3}
	\cnode[linewidth=1pt](3,0.5){6pt}{3.2}

	\nccurve[angleA=5,angleB=-130]{->}{1}{2}
	\nccurve[angleA=5,angleB=130]{->}{1}{3}
	\nccurve[angleA=-70,angleB=190]{->}{1}{3}
	
	\psarc[]{-}(.85,2.06){5pt}{-30}{33}
	
	\rput(.6,2.3){{a}}
	\rput(-.1,1.4){{b}}
	
	\rput(2.55,2.6){2}
	\rput(2.6,1.4){8}
	\rput(2.15,0.15){1}
	
	\psline[]{->}(-.6,2.6)(-.2,2.2)
\end{pspicture}
\hspace{1.5cm}
\begin{pspicture}[showgrid=false](0,0.3)(4,4)
	\psset{arrowsize=5pt,nodesep=0pt,arrowlength=1,linewidth=1pt}
	\cnode[](.9,3){8pt}{1a}
	\cnode[](.9,1){8pt}{1b}
	\cnode[](3,3.5){8pt}{2}
	\cnode[linewidth=1pt](3,0.5){8pt}{3}
	\cnode[linewidth=1pt](3,0.5){6pt}{3.2}

	\cnode[linewidth=2pt,fillstyle=solid,fillcolor=black](0,2){4pt}{discrete}
	\rput(0.1,2.4){a}
	\rput(.3,2.7){1}
	\rput(0.1,1.6){b}
	\rput(.3,1.3){1}
	\ncline[]{->}{discrete}{1a}
	\ncline[]{->}{discrete}{1b}

	\nccurve[angleA=-20,angleB=-135]{->}{1a}{2}
	\nccurve[angleA=-20,angleB=110]{->}{1a}{3}
	\nccurve[angleA=-40,angleB=190]{->}{1b}{3}
	
	\psarc[]{-}(1.65,2.76){6pt}{-60}{20}
	
	\rput(1.6,3.2){{a}}
	\rput(1.6,.9){{b}}
	
	\rput(2.7,2.75){2}
	\rput(3,1.4){8}
	\rput(2.15,0.15){1}
	
	\psline[]{->}(-.45,2.35)(-.1,2.1)
\end{pspicture}
\caption{An informal example, depicting the idea of the encoding of an early scheduling CTMG (left) in a late scheduling CTMG (right).}%
\label{fig:earlyToLate}%
\end{figure}

\end{document}